\pdfoutput=1
\def\passOptions#1#2{\PassOptionsToPackage{#2}{#1}}

\passOptions{xcolor}{table,dvipsnames,svgnames}

\documentclass[acmsmall,screen,nonacm]{acmart}

\usepackage{\jobname}

\AtBeginDocument{%
	\providecommand\BibTeX{{%
			\normalfont B\kern-0.5em{\scshape i\kern-0.25em b}\kern-0.8em\TeX}}}

\setcopyright{acmcopyright}
\acmPrice{}
\acmDOI{10.1145/3485514}
\acmYear{2021}
\copyrightyear{2021}
\acmSubmissionID{oopsla21main-p170-p}
\acmJournal{PACMPL}
\acmVolume{5}
\acmNumber{OOPSLA}
\acmArticle{137}
\acmMonth{10}

\begin{document}
	
	\title[Subtyping Machines]{Study of the Subtyping Machine \\ of Nominal Subtyping with Variance (full version)}
	
	\author{Ori Roth}
	\affiliation{%
		\institution{Technion---IIT}
		\city{Haifa}
		\country{Israel}}
	\email{soriroth@cs.technion.ac.il}
	
	\begin{abstract}
		This is a study of the computing power of the subtyping machine behind
Kennedy and Pierce's nominal subtyping with variance.
We depict the lattice of fragments of Kennedy and Pierce's type system and characterize
their computing power in terms of
regular, context-free, deterministic, and non-deterministic tree languages.
Based on the theory, we present Treetop---a generator of \NonCitingUse{CSharp}
implementations of subtyping machines.
The software artifact constitutes the first feasible (yet POC) fluent API generator
to support context-free API protocols in a decidable type system fragment.

	\end{abstract}

	
	
	\maketitle
	
	\section{Introduction}
	\label{section:aa}
	Object-oriented programming is all about inheritance or subtyping, and how
software can and cannot be properly engineered using it. In the eyes of the
theoretical computer scientist or the type theorist, the question is what
(compile-time) computations can and cannot be done with subtyping. The question
becomes challenging when generics, i.e., parametric polymorphism, interact with
subtyping, raising the question of when a certain generic taking type
argument~$t_1$ is a subtype of the same generic taking type argument~$t_2$. In the
invariant regime of subtyping with generics, the answer is negative (unless~$t_1=t_2$);
in the covariant regime, the answer is positive precisely when~$t_1$ is a subtype
of~$t_2$; in contrast, in the contravariant regime, when~$t_2$ is a subtype of~$t_1$.

Consider, for example, the~\CSharp code in \cref{lst:example:cff:program}
(unless said otherwise, all code examples in this paper are written in \CSharp):
\begin{code}[style=csharp,caption={\protect\CSharp subtyping machine recognizing palindromes over~$❴a, b❵$},
	label={lst:example:cff:program}]
interface a<out x> {} interface b<out x> {} interface E {}
class v0<x> : a<v0<a<x>>>, a<a<x>>, a<x>, b<v0<b<x>>>, b<b<x>>, b<x> {}
a<b<b<a<b<b<a<E>>>>>>> w1 = new v0<E>(); // Compiles: @\color{comment}$abbabba$@ is a palindrome
a<b<b<a<b<a<a<E>>>>>>> w2 = new v0<E>(); // Does not compile: @\color{comment}$abbabaa$@ is not a palindrome
\end{code}
The program in \cref{lst:example:cff:program} defines four classes (lines 1--2):
generic (polymorphic) classes \cc{a} and \cc{b} employing a covariant type parameter \cc{x},
as indicated by the keyword \kk{out};
non-generic class \cc{E};
and generic class \cc{v0} that inherits from the six supertypes written after the colon.
Below the classes are two variable declarations (lines 3--4).
The type of variable \cc{w1} is \cc{a<b<b<a<b<b<a<E>{}>{}>{}>{}>{}>{}>}, encoding the word $w₁=abbabba$,
and the type of variable \cc{w2} is \cc{a<b<b<a<b<a<a<E>{}>{}>{}>{}>{}>{}>}, encoding the word $w₂=abbabaa$.
\cref{lst:example:cff:program} demonstrate the encoding of~$ℓ$, the non-deterministic
context-free formal language of palindromes over alphabet~$❴a, b❵$, in the type
system of~\CSharp.
Each variable is assigned an instantiation of type \cc{v0<E>}, that
type checks if and only if the variable type encodes a palindrome.
Thus, the assignment to variable \cc{w1} type checks since~$w₁∈ℓ$,
but the assignment to \cc{w2} fails since~$w₂∉ℓ$.
The mechanics behind this ‟magic”, and why it matters, is
revealed by our study of the computing power of subtyping in the
context of the work of \citet{Kennedy:Pierce:07} on the \emph{nominal
subtyping with variance} type system.

\subsection{Kennedy and Pierce's Type System and Its Decidable Fragments}

\cref{figure:lattice} places our results in context.
The figure shows a three-dimensional lattice, structured similarly
to the~$λ$-cube~\cite{Barendregt:91}.
The lattice is spanned by three type system features, which are
explained briefly here and in greater detail in \cref{section:subtyping}.
The head of the lattice \TKP is the abstract type system developed by \citet{Kennedy:Pierce:07},
which employs all three features. The other lattice points are type system
fragments of \TKP that use different feature combinations.
Nodes are adorned with references to previous works (shown as citations) and to our
contributions (shown as forward references).

\colorlet{dregf}{Fuchsia}
\colorlet{regf}{Cyan}
\colorlet{dcff}{SpringGreen}
\colorlet{cffgnf}{Yellow}
\colorlet{cff}{YellowOrange}
\colorlet{ref}{OrangeRed}
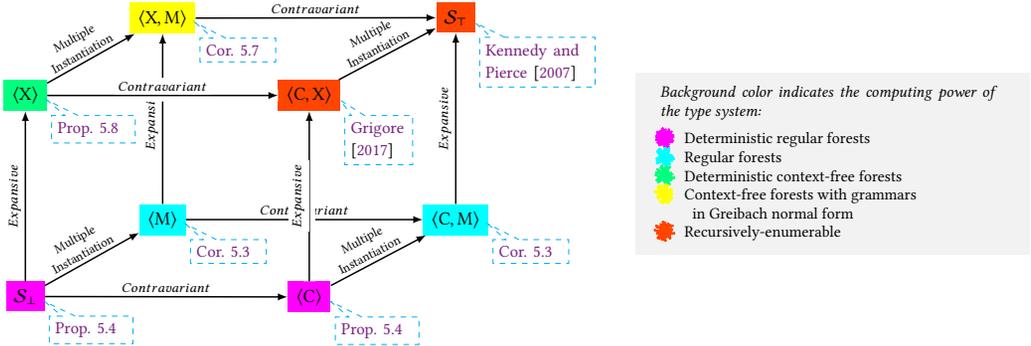
\begin{figure}[ht]
  \crefname{proposition}{Prop.}{Props.}
  \crefname{corollary}{Cor.}{Cors.}
  \centering
  \vspace{-5ex}
  \begin{tabular}{cc}
    \hspace{3ex}
    \begin{minipage}{.45\textwidth}
      \vspace{4ex}
      \adjustbox{scale=.7,center}{
        \begin{tikzcd}[row sep=2.5em,column sep=4.5em,execute at begin picture={
          \tikzstyle{aa}=[->,>=latex,thick,text=black]
          \tikzstyle{cc}=[aa]
          \tikzstyle{xx}=[aa,sloped]
          \tikzstyle{mm}=[aa,sloped,pos=.45]
          \tikzstyle{tt}=[rectangle,fill=##1]
          \tikzstyle{dec}=[tt=LightGreen!50]
          \tikzstyle{undec}=[tt=ref]
          \tikzstyle{reference}=[rectangle callout,draw=cyan,dashed,callout absolute pointer={(##1.base east)},below right=.5ex and .5ex of ##1]
          },execute at end picture={
              \node[text width=13ex,reference=cxm] {\small\citet{Kennedy:Pierce:07}};
              \node[text width=7ex,reference=cx] {\small\citet{Grigore:2017}};
              \node[text width=9ex,reference=bot] {\small\cref{theorem:dregf}};
              \node[text width=9ex,reference=c] {\small\cref{theorem:dregf}};
              \node[text width=8ex,reference=m] {\small\cref{theorem:regf}};
              \node[text width=8ex,reference=cm] {\small\cref{theorem:regf}};
              \node[text width=9ex,reference=x] {\small\cref{theorem:dcff}};
              \node[text width=8ex,reference=xm] {\small\cref{theorem:cff:gnf}};
            }]
          & |[alias=xm,dec,fill=cffgnf]| \Txm & & |[alias=cxm,undec]| \TKP ⏎
          |[alias=x,dec,fill=dcff]| \Tx & & |[alias=cx,undec]| \Tcx ⏎
 ⏎
          & |[alias=m,dec,fill=regf]| \Tm & & |[alias=cm,dec,fill=regf]| \Tcm ⏎
          |[alias=bot,dec,fill=dregf]|\Tbot & & |[alias=c,dec,fill=dregf]|\Tc ⏎
          \ar[from=bot,to=c,cc]{}{Contravariant}
          \ar[from=bot,to=x,xx]{}[anchor=south]{Expansive}
          \ar[from=bot,to=m,mm]{}{\shortstack{\scriptsize Multiple ⏎ \scriptsize Instantiation}}
          \ar[from=c,to=cm,mm]{}{\shortstack{\scriptsize Multiple ⏎ \scriptsize Instantiation}}
          \ar[from=x,to=xm,mm]{}{\shortstack{\scriptsize Multiple ⏎ \scriptsize Instantiation}}
          \ar[from=m,to=cm,cc]{}{Contravariant}
          \ar[from=m,to=xm,xx]{}[anchor=south]{Expansive}
          \ar[from=xm,to=cxm,cc]{}{Contravariant}
          \ar[from=cx,to=cxm,mm]{}{\shortstack{\scriptsize Multiple ⏎ \scriptsize Instantiation}}
          \ar[from=cm,to=cxm,xx]{}[anchor=south]{Expansive}
          \ar[from=c,to=cx,xx,crossing over]{}[fill=white,anchor=south]{Expansive}
          \ar[from=x,to=cx,cc,crossing over]{}[fill=white]{Contravariant}
        \end{tikzcd}
      }
    \end{minipage} \hspace{6ex} &
    \begin{adjustbox}{margin=1ex,minipage=.355\linewidth,bgcolor=black!5}\tiny
		\begin{tabular}{l@{\hskip 6pt}l}
			\multicolumn{2}{p{.9\linewidth}}{\emph{Background color indicates the computing power of the type system:}} \\[4ex]
			\raisebox{-1pt}{\color{dregf}\splash} & Deterministic regular forests \\
			\raisebox{-1pt}{\color{regf}\splash} & Regular forests \\
			\raisebox{-1pt}{\color{dcff}\splash} & Deterministic context-free forests \\
			\raisebox{-1pt}{\color{cffgnf}\splash} & Context-free forests with grammars \\
			                                       & \enskip in Greibach normal form \\
			\raisebox{-1pt}{\color{ref}\splash} & Recursively-enumerable
		\end{tabular}
    \end{adjustbox}
  \end{tabular}
  \vspace{-5ex}
  \caption{
    Lattice describing the eight subtyping type system fragments employing
    contravariance ({\normalfont C}), expansive inheritance ({\normalfont X}), and multiple
  instantiation inheritance ({\normalfont M}) (defined in \protect\cref{section:subtyping})}
  \label{figure:lattice}
\end{figure}

The lattice is spanned by three orthogonal Boolean coordinates or features
denoted by~$\text{C}$,~$\text X$, and~$\text M$, representing the three
core features of \TKP: \emph{\textbf Contravariance}, \emph{e\textbf
Xpansively-recursive inheritance}, and \emph{\textbf Multiple instantiation
inheritance}. The most familiar feature is the kind of variance;
the horizontal direction in the figure
denotes the distinction between covariant type systems (positioned
on the plane on the left-hand side of the lattice) and contravariant
type systems (right-hand).

Precise definitions of the three features follow in \cref{section:subtyping}; here we provide a quick
review of the concepts and the intuition behind them:
\begin{description}
  \item[(C)] \emph{Contravariance}:
    Type arguments assigned to contravariant type parameters can be substituted with less derived types.
    For example, using the ``\emph{is a}'' relationship we can say that if \cc{x} in \cc{a<x>} is contravariant and \cc{b} \emph{is a} \cc{c}
    then \cc{a<c>} \emph{is an} \cc{a<b>}.
    \begin{code}[style=csharp]
interface a<in x> {} // type parameter $\color{comment}\cc{x}$ of class $\color{comment}\cc{a}$ is contravariant
interface b : c {}   // type $\color{comment}\cc{b}$ is a subtype of $\color{comment}\cc{c}$
a<b> x=(a<c>) null;  // type $\color{comment}\cc{a<c>}$ is a subtype of $\color{comment}\cc{a<b>}$
\end{code}
  \item[(X)] \emph{Expansively-recursive inheritance}:
    The expansion of types through inheritance (as defined by \citet{Viroli:00}) is unbounded.
    Intuitively, this expansion involves the \emph{curiously recurring template pattern} \cite{Coplien:1996},
    where a class appears in one of its supertypes' generic arguments.
    \begin{code}
interface a<x, y> : b<a<b<x>, y>> {} // class $\color{comment}\cc{a}\!$'s inheritance is expansively-recursive
\end{code}
  \item[(M)] \emph{Multiple instantiation inheritance}:
    Inheritance is non-deterministic \cite{Kennedy:Pierce:07}, i.e., type arguments of supertypes are not unique.
    \begin{code}
interface a : b<c>, b<d> {} // supertype $\color{comment}\cc{b}$ is instantiated twice, with arguments $\color{comment}\cc{c}$ and $\color{comment}\cc{d}$
\end{code}
\end{description}

A point in the lattice is denoted by the set of its features. Type
system~$\TKP$ is denoted by the triple~$\Tcxm$, and the notation~$\Tcx$ designates
a type system that allows contravariance and expansive inheritance, but not
multiple instantiation inheritance. The code examples above also demonstrate
two key features ingrained within \TKP and its fragments: ‟polyadic” parametric
polymorphism and multiple inheritance, allowing a class to have multiple type
parameters and supertypes, respectively.

As shown in \cref{figure:lattice}, the new results here, together with the proof that \TKP
is undecidable due to Kennedy and Pierce and the subsequent proof
by~\citet{Grigore:2017} implying that \Tcx is undecidable, completely characterize
the computing power of all eight points of the lattice.
This characterization is not in terms of the Chomsky hierarchy of classes of
formal languages, but rather in terms of a similar hierarchy of classes of
forests, also called formal tree languages (see, e.g.,~\cite{Comon:07}), as
depicted by \cref{figure:hierarchy}.
The classes appearing in the figure are defined in \cref{section:forests}.
\begin{figure}[ht]
  \hspace{0ex}
  \begin{minipage}{.2\textwidth}\footnotesize
    \adjustbox{scale=.8,center}{
      \begin{tikzpicture}[node distance=5ex]
        \node[fill=dregf] (drf) {DRF};
      \node[above left=of drf,fill=regf] (nrf) {NRF};
        \node[above right=of drf,fill=dcff] (dcf) {DCF};
        \node[above=10ex of drf,fill=cffgnf] (cffgnf) {NCF\sub{GNF}};
        \node[above=of cffgnf,fill=cff] (cff) {NCF};
        \node[above=of cff,fill=ref] (re) {REF};
        \draw (drf)--(nrf);
        \draw (drf)--(dcf);
        \draw (nrf)--(cffgnf);
        \draw (dcf)--(cffgnf);
        \draw (cffgnf)--(cff);
        \draw (cff)--(re);
      \end{tikzpicture}
    }
  \end{minipage}
  \hspace{6ex}
  \begin{adjustbox}{margin=.5ex 1ex -7.5ex 1ex,minipage=.6\linewidth,bgcolor=black!5}\tiny
  	\begin{tabular}{l@{\hskip 6pt}l@{\hskip 6pt}l}
      \raisebox{-1pt}{\color{dregf}\splash} & \textbf{DRF} & \textbf Deterministic \textbf Regular \textbf Forests \\
      \raisebox{-1pt}{\color{regf}\splash} & \textbf{NRF} & \textbf Non-deterministic \textbf Regular \textbf Forests \\
      \raisebox{-1pt}{\color{dcff}\splash} & \textbf{DCF} & \textbf Deterministic \textbf Context-free \textbf Forests \\
      \raisebox{-1pt}{\color{cffgnf}\splash} & \multirow[t]{2}{*}{\textbf{NCF\sub{GNF}}} & \textbf Non-deterministic \textbf Context-free
        \textbf Forests with \\
        && \enskip grammars in \textbf Greibach \textbf Normal \textbf Form \\
      \raisebox{-1pt}{\color{cff}\splash} & \textbf{NCF} & \textbf Non-deterministic \textbf Context-free \textbf Forests \\
      \raisebox{-1pt}{\color{ref}\splash} & \textbf{REF} & \textbf Recursively-\textbf Enumerable \textbf Forests
    \end{tabular}
  \end{adjustbox}
  \caption{A Chomsky-like hierarchy of formal forest classes (defined in \cref{section:forests})}
  \label{figure:hierarchy}
\end{figure}
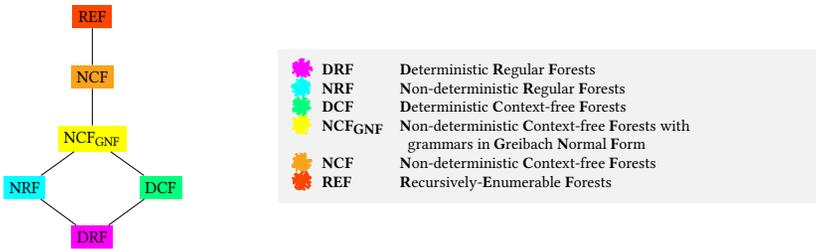

As shown in \cref{figure:hierarchy}, the hierarchy of tree
languages is a non-linear lattice, spanned by two orthogonal Boolean features
abbreviated by letters in the acronyms in the figure as follows:
\begin{itemize}
  \item \emph{Determinism} ($\text D<\text N$) distinguishing between
    \emph{\textbf Deterministic} forests and their super-set,
    \emph{\textbf Non-deterministic} forests;
  \item \emph{Method of specification} ($\text R<\text C$) distinguishing between \emph{\textbf
    Regular} forests and their super-set, \emph{\textbf Context-free} forests.
\end{itemize}
In \cref{figure:hierarchy}, we see that forest class DRF is the smallest class in the diagram.
Moving up the lattice, the two next classes, NRF and DCF, cannot be compared. Both are contained
in NCF, which is contained in REF, the class of all recursively enumerable
forests. We also make use of class~$\text{NCF\sub{GNF}}⊂\text{NCF}$, which
includes the non-deterministic context-free forests that can be specified by
a tree grammar in Greibach normal formal (GNF). Observe that class NCF\sub{GNF} also
contains DCF and NRF.

The theoretical contribution of our work is a
characterization of the computing power of the six decidable points in the
lattice of \cref{figure:lattice} in terms of the hierarchy
of~\cref{figure:hierarchy}.
These results are obtained by examining the \emph{subtyping machine} behind
subtyping queries---a kind of automata with finite control (dictated by the
nominal inheritance table) whose auxiliary storage is not a tape, stack, or
tree, but rather the contents of a specific subtyping query that changes along the
run of this automaton. These machines stand behind Kennedy and Pierce's study and Grigore's
emulation.

\subsection{Fluent APIs and Treetop}
\label{section:fluent}

Fluent API (application programming interface)~\cite{Fowler:2005} is a trendy
practice for the creation of software interfaces~\cite{Nakamaru:2020b}. In a
statically-typed, object-oriented setting, a fluent API method returns a
type that receives other API methods. In this way, fluent API methods are invoked
in a chain of consecutive calls:
\begin{equation}\label{equation:fluent:api}
  \cc{$σ₁$().$σ₂$().$⋯$.$σₙ$()}
\end{equation}
The main advantage of fluent API is its ability to enforce a protocol at
compile time. If an invocation of method~$σ_{i+1}$ after a call to~$σᵢ$ breaks
the protocol, then the type returned from method~$σᵢ$ must not
declare~$σ_{i+1}$, which would cause the expression to fail type checking.

A practical application of fluent API is embedding domain specific languages
(DSLs) in a host programming
language~\cite{Xu:2010,Gil:Levy:2016,Nakamaru:17,Gil:19,Yamazaki:2019,Nakamaru:2020}.
The syntax and semantics of a DSL are optimally designed for its domain; for
instance, \SQL is a DSL suited for composing database queries. If the protocol
of the fluent API is a grammar~$G$ describing the DSL, then the chain of method
calls of \cref{equation:fluent:api} type checks if and only
if~$w=σ₁σ₂⋯σₙ∈L(G)$, i.e., the expression encodes a legal DSL program. For
example, with an \SQL fluent API implemented in \Scala we can embed database
queries as \Scala expressions:
\[
\text{\inline[language=scala]{select("id").from(grades).where(_.grade>90)}.}
\]
The expression
\[
\text{\inline[language=scala]{select("name").where(_.length==4)},}
\]
however,
does not type check, as it represents an illegal \SQL query,
\[
\text{\inline[language=sql]{SELECT name WHERE length=4}.}
\]

The academic study of fluent APIs seeks new ways to encode increasingly complex
DSL classes. Results in the field are given in the form of \emph{generators}:
software artifacts that convert a formal grammar into a fluent API encoding its
language. Fling~\cite{Gil:19} and TypeLevelLR~\cite{Yamazaki:2019} are two
prominent fluent API generators that support all deterministic context-free
languages (DCFLs). \citet{Grigore:2017} implemented a generator that accepts
context-free grammars describing general context-free DSLs---a proper
super-set of DCFLs. Specifically, his implementation coded a CYK parser
in an intermediate simple, yet high level language, and then translated
the code in this language into a Turing machine, being emulated in a
subtyping query. Unfortunately, the resulting APIs are
unusable in practice, as they lead to extremely long compilation times and
crash the compiler.

Examining \cref{figure:lattice} we see that the \emph{contravariant} regime for
combining genericity with subtyping is essential for the undecidability
result of Kennedy and Pierce and its use by Grigore to encode Turing machines with \Java generics.
Contravariance is also essential to Grigore's implementation of a
fluent API generator for \emph{all} context-free languages.
Grigore's result is noteworthy since the best previous fluent API generators
\linebreak\cite{Gil:Levy:2016,Gil:19,Yamazaki:2019}, relying on invariance,
were only able to encode the deterministic subset of context-free
languages (languages with an LR grammar). As shown in \cref{figure:lattice}, our results
prove that the covariant regime, exemplified by the simple phrase ‟\emph{a bag
of apples is bag of fruit}”, is weaker than its contravariant counterpart.

Establishing that covariance is weaker than contravariance and noticing that it is
stronger than invariance, it is natural to ask whether the entire class of
context-free languages can be encoded in a type system with nominal subtyping
combined with only covariant genericity. Conversely, the question is motivated
by the argument that there is no theoretical difficulty in recognizing context-free
languages in a Turing-complete type system, but it might be a challenge
to recognize these in a decidable type system.

Unfortunately, the code produced by Grigore's fluent API generator collapses
under its own weight (measured in gigabytes) and crashes rugged modern
compilers. Another natural question to ask is: Is the (computationally)
weaker artillery of covariance also lighter in weight (in terms of code size)?

Treetop is a product of our theoretical study that answers both questions
affirmatively. It is a proof of concept (POC) generator of ‟covariant generic
subtyping” machines that are implemented as \NonCitingUse{CSharp} code. 
Given $G$, a CFG defining a formal language~$ℓ$ (such as the language
of palindromes in the example of \cref{lst:example:cff:program}), 
Treetop generates a library in \CSharp
that defines a fluent API of~$ℓ$. 
The library uses expansive inheritance (the X feature in the lattice of
\cref{figure:lattice}) and multiple instantiation inheritance (M in this
lattice), but not contravariance~(C).
Thus, Treetop not only supports general context-free grammars,
but it also does this in a decidable type system fragment.
We demonstrate in an experiment non-trivial subtyping machines
generated by Treetop that compile relatively fast, yet we
also identify at least one example which induces impractically long
compilation times using the CSC \CSharp compiler.

	\paragraph*{Outline}
	\cref{section:forests,section:subtyping} elaborate on
	Kennedy and Pierce's work on subtyping, to refresh readers' understanding of concepts such as formal forests and
	tree grammars, and set up our notations. Subtyping machines are then defined
	in \cref{section:machines} as computational devices that model nominal
	subtyping. \cref{section:expressiveness} derives the results of
	\cref{figure:lattice}. Treetop, our API generator for context-free protocols,
	is the focus of \cref{section:generator}. \cref{section:zz} discusses the
	results of our study.
	
	\paragraph*{Note on bold notation.}
	The bold version of a symbol denotes a (possibly empty) sequence of instances of the
	non-bold version of this symbol, e.g.,~$\vv{x}$ stands for~$x₁,x₂,…,xₖ$ for
	some~$k≥0$.
	This notation is abused when the intent is clear from the context, e.g., \[
		\vv{τ}[\vv{x}←\vv{t}]
	\] abbreviates \[
		τ₁[x₁← t₁,…,xₖ← tₖ],…,τₙ[x₁← t₁,…,xₖ← tₖ],
	\] where~$\vv{τ}$,~$\vv{x}$, and~$\vv{t}$ are the three
	sequences~$\vv{τ}=τ₁…,τₙ$,~$\vv{t}=t₁…,tₙ$, and~$\vv{x}=x₁,…,xₙ$.
	
	\section{Preliminaries: Formal Forests and Their Grammars}
	\label{section:forests}
	We assume that the reader is familiar with the classical theory of automata and formal
languages (e.g.,~\cite{Hopcroft:Motwani:Ullman:07}). Here we provide a quick
review of the standard generalization of the theory to \emph{forests}, also
called (formal) \emph{tree languages}. (For a more detailed review, see standard references such as
\cite{Comon:07,Gecseg:97,Osterholzer:18}.)

We use parenthetical notation for terms (trees), e.g.,
\begin{equation}\label{eq:example:term}
  t=σ₁(σ₂(σ₃()),σ₄())
\end{equation}
Rank~$n∈ℕ$ of symbol~$σ$, denoted~$\rank(σ)$, is the number of its children.
If~$n=0$, then~$σ$ is called a leaf, and if~$n=1$, then~$σ$ is monadic.
In both cases,~$σ$'s parentheses may be omitted, e.g.,~$t$ in
\cref{eq:example:term} becomes~$σ₁(σ₂σ₃,σ₄)$.
The height of a tree~$t$ is denoted~$\height(t)$, e.g., for~$t$ in
\cref{eq:example:term},~$\height(t)=3$. The set of trees over ranked
alphabet~$Σ$ (i.e., an alphabet of symbols with ranks) is denoted~$Σ^▵$,
generalizing the Kleene closure~$K^*$ for forests.

Consider the following definitions, adapted from \citet[\S{}2.1.1]{Comon:07}:
\begin{equation}\label{eq:example:regular:grammar}
\begin{aligned}
  \text{(a)~}\-Nat-&\produce \+z+& \text{(b)~}\-Nat-&\produce \+s+(\-Nat-) ⏎
  \text{(c)~}\-List-&\produce \+nil+& \text{(d)~}\-List-&\produce \+cons+(\-Nat-, \-List-)
\end{aligned}
\end{equation}
These are tree grammar productions deriving Peano numbers (productions (a,b))
and cons-lists of those integers (productions (c,d)).
An integer~$n∈ℕ$ is encoded by the monadic term~$\+s+ⁿ\+z+$.
Instead of strings, tree grammar productions derive terms:
The left-hand side of a production is a variable, and the right-hand
side is a term with which the variable can be substituted.
As in the string case, terminal (ground) terms are derived from variables by
repeatedly substituting the latter according to the given productions.
For example, we encode the list~$⟨{}2,1,0⟩$ using the term
\begin{equation}\label{eq:example:list}
  \+cons+(\+ssz+, \+cons+(\+sz+, \+cons+(\+z+, \+nil+)))
\end{equation}
This term can be derived from variable \-List- from
\cref{eq:example:regular:grammar}, as illustrated in \cref{figure:derivation:example}.
\cref{figure:derivation:example} uses arrows to depict variable substitution; an arrow label
references the production from \cref{eq:example:regular:grammar} used to apply the derivation.
\begin{figure}
  \def\-#1-{{\textit{\textsf{#1}}}}
  \begin{tikzpicture}
    \tikzstyle{e}=[->,draw=black]
    \tikzstyle{l}=[left,midway,font=\tiny,scale=.75,color=black!60,yshift=-.5pt]
    \tikzstyle{l2}=[l,below left=-3pt,yshift=1pt]

    \node (l0) {$\-\underline{List}-$};
    \node (l1) [below=1ex of l0] {$\+cons+(\-\underline{Nat}-, \-\underline{List}-)$};
    \path[e] ($(l0.south)+(0pt,4pt)$)--($(l1.north)+(0pt,-2pt)$) node[l] {$(d)$};
    \node (l2) [below right=1ex and 6ex of l1.south] {$\+cons+(\-\underline{Nat}-, \-\underline{List}-)$};
    \path[e] ($(l1.south east)+(-15pt,4.5pt)$)--($(l2.north west)+(13pt,-3pt)$) node[l2] {$(d)$};
    \node (l3) [below right=1ex and 6ex of l2.south] {$\+cons+(\-\underline{Nat}-, \-\underline{List}-)$};
    \path[e] ($(l2.south east)+(-15pt,4.5pt)$)--($(l3.north west)+(13pt,-3pt)$) node[l2] {$(d)$};
    \node (l4) [below right=1.75pt and 6ex of l3.south] {$\+nil+$};
    \path[e] ($(l3.south east)+(-15pt,4.5pt)$)--($(l4.north)+(-3pt,-2pt)$) node[l2] {$(c)$};

    \node (n00) [below=1ex and 3ex of l1.south] {$\+s+(\-\underline{Nat}-)$};
    \path[e] ($(l1.south)+(0pt,4.5pt)$)--($(n00.north)+(0pt,-1pt)$) node[l] {$(b)$};
    \node (n01) [below right=1ex and .5ex of n00.south west,anchor=north west] {$\+s+(\-\underline{Nat}-)$};
    \path[e] ($(n00.south)+(.5ex,4.5pt)$)--($(n01.north)+(0pt,-1pt)$) node[l] {$(b)$};
    \node (n02) [below right=1ex and-3.5pt of n01.south] {$\+z+$};
    \path[e] ($(n01.south)+(.5ex,4.5pt)$)--($(n02.north)+(0pt,-1pt)$) node[l] {$(a)$};
    \node (n10) [below=1ex and 3ex of l2.south] {$\+s+(\-\underline{Nat}-)$};
    \path[e] ($(l2.south)+(0pt,4.5pt)$)--($(n10.north)+(0pt,-1pt)$) node[l] {$(b)$};
    \node (n11) [below right=1ex and-3.5pt of n10.south] {$\+z+$};
    \path[e] ($(n10.south)+(.5ex,4.5pt)$)--($(n11.north)+(0pt,-1pt)$) node[l] {$(a)$};
    \node (n20) [below=1ex and 3ex of l3.south] {$\+z+$};
    \path[e] ($(l3.south)+(0pt,4.5pt)$)--($(n20.north)+(0pt,-1pt)$) node[l] {$(a)$};
  \end{tikzpicture}
  \caption{The derivation of the term in \protect\cref{eq:example:list}, encoding
  	the list~$⟨2,~1, 0⟩$, from variable \textit{\textsf{List}} according to the
  	productions of \protect\cref{eq:example:regular:grammar}}
  \label{figure:derivation:example}
\end{figure}
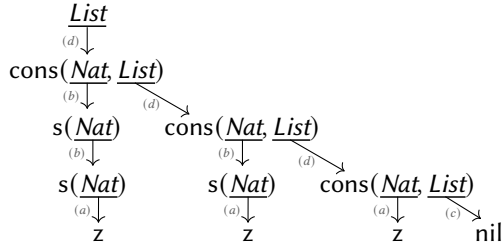

For readers familiar with functional programming languages as
\ML and \Haskell, it is intuitive to see the grammar in \cref{eq:example:regular:grammar}
as a set of datatype definitions, as shown in \cref{lst:sml:regular:grammar}.
The grammar variables \-Nat- and \-List- are datatypes whose value
constructors are the grammar terminals \+z+, \+s+, \+nil+, and \+cons+.
For each grammar production $v \produce \sigma$, we define $\sigma$ as a
value constructor of $v$;
if $\sigma$ has children, they are used as constructor parameters.
Type checking in the resulting \ML program coincides with derivation of
Peano numbers.
In particular, an \ML expression belongs to type \cc{List} if and only
if it can be derived from variable \-List-, e.g., the expression in line 3
of \cref{lst:sml:regular:grammar} encodes the tree in \cref{eq:example:list}
and compiles as expected.

\begin{code}[language=ml,caption={The grammar of \cref{eq:example:regular:grammar} encoded as~\ML datatypes},label={lst:sml:regular:grammar}]
datatype Nat = z | s of Nat;
datatype List = nil | cons of Nat * List;
cons(s(s(z)), cons(s(z), cons(z, nil))) : List;
\end{code}

Formally, a forest grammar~$G=⟨Σ, V, v₀, R⟩$ employs a finite ranked alphabet~$Σ$, a
finite ranked set of variables~$V$, a designated initial leaf~$v₀∈V$, and a
finite set of productions~$R$. To generate a terminal tree, grammar~$G$ repeatedly
applies~$R$ productions to the initial leaf~$v₀$, as explained below.
Intermediate trees, members of~$(Σ∪V)^▵$, are called tree forms, similar to
sentential forms of string grammars.

Let~$X$ be a finite set of \emph{parameter} leaf symbols, disjoint to sets~$Σ$
and~$V$. Terms over~$Σ∪V$, denoted~$t$, are ground, while terms over~$Σ∪V∪X$, denoted~$τ$, are unground. Unground terms are used as tree
patterns, where the~$X$ parameters can be substituted with terms (usually
ground). The application of substitution~$s=❴\vv{x}←\vv{t} ❵=❴x₁← t₁,…,xₖ←
tₖ ❵$ to unground term~$τ$, denoted~$τ[s]$, yields a term where every
occurrence of parameter~$xᵢ$ is substituted with term~$tᵢ$.
(A substitution's curly brackets may be omitted if written inside square brackets.)
A term~$τ$ matches term~$τ'$ if and only if there exists a substitution~$s$ such
that~$τ'[s]=τ$, denoted~$τ⊑ₛτ'$.
Let~$\params(τ)⊆X$ denote the set of parameters appearing in term~$τ$.

A production~$ρ=τ \produce τ'∈R$ describes a tree transformation.
If tree~$t$ matches~$τ$ with substitution~$s$,~$t⊑ₛτ$, the
application of~$ρ$ to~$t$, denoted~$t[ρ]$, yields~$τ'$ substituted by~$s$, \[
  t⊑ₛτ~⇒~t[τ\produceτ']=τ'[s].
\] We require that every parameter~$x∈X$ on the right-hand side of production~$ρ∈R$
be found on the left-hand side, so that all the applications of~$ρ$ to ground trees are ground.

Note the distinction between variables and parameters.
Variables are found in tree forms and constitute the nodes remaining to be derived.
Parameters, on the other hand, appear only in productions and are used as
templates in tree patterns.

We say that tree form $t$ derives tree form $t'$, denoted $t \rightarrow t'$ (or $t \rightarrow_G t'$),
if $t$ can be transformed into $t'$ by a single application of an $R$ production: \[
	t \rightarrow t' ~\Leftrightarrow~ \exists \rho \in R, ~ t[\rho] = t'.
\]
The tree language (forest) of grammar $G$, $L(G)$, is the set of terminal trees that can be
derived from the initial variable $v_0$ in any number of steps: \[
	L(G) = \{ t \in \Sigma^\triangle ~|~ v_0 \rightarrow^*_G t \}.
\]

Each class (family) of tree grammars is defined by a set of restrictions imposed on the
shape of productions.
In regular tree grammars that define regular forests (NRF), productions are
of the form~$v \produceσ(\vv{v})$, where~$σ∈Σ$ and~$v,\vv{v}∈V$; a more relaxed definition simply restricts
variables to be leaves.
Observe that the productions of \cref{eq:example:regular:grammar} belong to
a regular grammar.

Productions of context-free tree grammars (CFTG), yielding context-free forests (NCF),
are of the form~$v(\vv{x}) \produceτ$ (the parameters in~$\vv{x}$ are unique).
For example, the following productions belong to a CFTG:
\begin{equation}\label{eq:example:derivation:cf}
\text{(a)~}v₀\produce v₁(\+leaf+) \quad \text{(b)~}v₁(x)\produce v₁(\+node+(x, x)) \quad \text{(c)~}v₁(x)\produce x
\end{equation}
These productions derive complete binary trees over~$❴\+leaf+, \+node+❵$.
The binary trees are derived recursively.
The child of variable~$v₁$, captured by parameter~$x$, describes a tree
of height~$n≥1$; production (b) is then used to replace~$x$ with~$\+node+(x, x)$,
the complete binary tree of height~$n+1$.
The recursion terminates by production (c), which simply returns the constructed tree.

To encode context-free tree grammars in \ML, we use polymorphic datatypes,
i.e., datatypes that employ type parameters.
Consider, for example, the following grammar:
\begin{equation}\label{eq:context:free:grammar:ml}
\begin{aligned}
	v_0(x) &\produce b(x) & v_1 &\produce E \\
	v_2(x) &\produce m(x) & v_2(x) &\produce a(v_2(v_0(x)))
\end{aligned}
\end{equation}
This grammar defines the context-free language (forest)~$a^n m b^n E$ ($n \ge 0$).
The \ML program encoding the grammar is shown in \cref{lst:sml:context:free:grammar}.
Variables $v_0$ and $v_2$, which use parameters, are encoded by
polymorphic datatypes that define the parameterized type constructors
\cc{v0} and \cc{v2}.
Following the datatype definitions are three expressions:
While the expression $a^3 m b^3 E$ (line 4) compiles,
the expressions $a^1 m b^2 E$ and $a^2 m b^1 E$ (lines 5--6)
do not, as expected.

\begin{code}[float=ht,language=ml,caption={The grammar of \cref{eq:context:free:grammar:ml} encoded as polymorphic~\ML datatypes},label={lst:sml:context:free:grammar}]
datatype 'x v0 = b of 'x;
datatype v1 = E;
datatype 'x v2 = m of 'x | a of 'x v0 v2;
a(a(a(m(b(b(b(E))))))) : v1 v2;
a(m(b(b(E)))) : v1 v2; (* Does not compile *)
a(a(m(b(E)))) : v1 v2; (* Does not compile *)
\end{code}

Not all context-free grammars, however, can be encoded in \ML using this method.
A production deriving a grammar variable to one of its parameters, e.g.,
production (c) in \cref{eq:example:derivation:cf}, results in a nonsensical
\ML declaration \inline[language=ml]{datatype 'x v='x}.
In addition, a grammar terminal that appears in more than one production
introduces overloaded value constructors, that not supported by \ML.

In CFTG in Greibach normal form (GNF), the roots of the right-hand side of productions are
terminal,~$v(\vv{x}) \produceσ(\vv{τ})$~\cite{Guessarian:83}.
Observe that none of the productions in \cref{eq:example:derivation:cf} are in GNF:
The roots of the right-hand sides are variables in productions (a,b) and a parameter
in production (c), both forbidden in GNF.
In contrast, the following productions do belong to a CFTG in GNF:
\begin{equation}\label{eq:example:cff:grammar}
\begin{array}{ccc}
	\text{(a)~}v₀x \produce av₀ax & \text{(b)~}v₀x \produce aax & \text{(c)~}v₀x \produce ax ⏎
	\text{(d)~}v₀x \produce bv₀bx & \text{(e)~}v₀x \produce bbx & \text{(f)~}v₀x \produce bx
\end{array}
\end{equation}
Variable $v_0$ derives palindromes over $\{a,b\}$. For instance, the palindrome
$abbabba$ is derived from the tree form $v_0 E$ ($E$ is a dummy leaf) as follows: \[
	v_0E \xrightarrow{\text{(a)}} av_0aE \xrightarrow{\text{(d)}} abv_0baE \xrightarrow{\text{(d)}} abbv_0bbaE \xrightarrow{\text{(c)}} abbabbaE
\] (The production from \cref{eq:example:cff:grammar} used in each derivation is written above the arrow.)
The productions in \cref{eq:example:cff:grammar} are context-free, as $v_0$ employs a parameter,
and they are in GNF, as the root of the right-hand side of each
production is a terminal (either $a$ or $b$).

Notice that all regular grammars are context-free and in GNF.

In deterministic context-free tree grammars in GNF,
including deterministic regular tree grammars, if both~$ρ₁=v \produceσ(\vv{τ})$
and~$ρ₂=v \produceσ(\vv{τ}')$ are productions in~$R$, then~$\vv{τ}=\vv{τ}'$.
Deterministic regular grammars define deterministic regular forests (DRF), a proper
subset of regular forests, also called path-closed forests~\cite{Comon:07}.
Likewise, deterministic context-free tree grammars define deterministic
context-free forests (DCF).
In contrast to context-free string grammars, not every context-free tree grammar
has a GNF---but every deterministic tree grammar does~\cite{Guessarian:83}.
The class of context-free forests that can be described by a grammar in GNF
is denoted as NCF\sub{GNF}.

In this paper, we use an extended version of context-free tree grammars where
the initial variable~$v₀$ is replaced by an initial tree.
\begin{definition}[ECFTG]\label{definition:ecftg}
  Extended context-free tree grammars (ECFTGs) generalize standard CFTGs
  by allowing initialization by any tree~$t₀$. Otherwise, the semantics are the same.
\end{definition}
A standard CFTG is extended by definition, with~$t₀=v₀$.
We show that extended grammars are as expressive as standard grammars:
\begin{lemma}\label{lemma:ecftg}
  For every ECFTG (in GNF), there exists a CFTG (in GNF) that defines the same forest.
\end{lemma}
\begin{proof}
  Let $G=⟨Σ, V, t₀, R⟩$ be an ECFTG.
  If the initial tree~$t₀$ is a variable,~$t₀=v∈V$ and we are done.
  Otherwise, let CFTG~$G'=⟨Σ, V', v₀, R'⟩$,
  introducing the initial variable~$v₀∉V$,~$V'=V∪❴v₀ ❵$, and
  employing the original productions,~$R⊆R'$.
  If the initial tree is terminal,~$t₀=σ₀∈Σ$, simply add the rule~$v₀ \produce σ₀$ to~$R'$.
  Else~$t₀$ is not a leaf: Let~$t₀=ξ₀(\vv{t₀})$, where the root~$ξ₀∈Σ∪V$
  is of rank one or above.
  For every production~$ξ₀(\vv{x}) \produceτ∈R$ if~$ξ₀$ is a variable, or production~$ξ₀(\vv{x}) \produceτ$,~$τ=ξ₀(\vv{x})$ if~$ξ₀$ is terminal,
  add production~$v₀ \produceτ[\vv{x}←\vv{t₀}]$ to~$R'$.

  Derivations of the initial variable~$v₀$ expand it into the initial tree~$t₀$ and
  apply an original~$R$ derivation in a single step.
  After~$v₀$ is derived, the resulting tree form is identical to the one derived from~$t₀$, \[
    v₀[v₀ \produceτ[\vv{x}←\vv{t₀}]]=t₀[ξ₀(\vv{x}) \produceτ],
  \] thus they derive the same forest.

  This construction preserves GNF.
  Newly introduced rules derive tree~$τ[\vv{x}←\vv{t₀}]$,
  whose root is terminal given that either \1~the root of~$t₀$ is terminal,
  or otherwise \2~$τ$ is the right-hand side of a rule in~$R$, starting with a terminal
  given that the original grammar is in GNF.
\end{proof}

	\section{Preliminaries: Nominal Subtyping with Variance}
	\label{section:subtyping}
	This work focuses on \citets{Kennedy:Pierce:07} abstract type system of \emph{nominal subtyping with variance}.
This type system supports nominal subtyping with declaration-site variance
as found in, e.g., \CSharp and \Scala.
\cref{figure:type:system} describes the abstract syntax (sub-figures (a,b))
and semantics (sub-figures (c,d,e)) of the type system, both adopted from
\citet{Kennedy:Pierce:07}.
Note that we reuse notations introduced in \cref{section:forests}
in the context of formal forests, foreshadowing the connection
between the two formalisms.

\begin{figure}[t]
  \small
  \begin{adjustbox}{center,scale=.8}
    \setlength{\tabcolsep}{1ex}
    \begin{tabular}{|c|c|}
      \hline
      \renewcommand{\arraystretch}{1.1}
      \begin{tabular}{ll}
        $P \produceΔ~q$ & ⏎
        $Δ\produceδ₁~…~δₙ$ & \constraint{$n≥0$} ⏎
        $δ\produce γ(\vv{◇ x}):\vv{τ}$ &
          \!\!\!\!\constraint{$\begin{cases}\params(\vv{τ})⊆\vv{x} ⏎\vv{τ}∩\vv{x}=∅\end{cases}$} ⏎
        $\vv{x} \produce◇ x₁,…,◇ xₖ$ &
          \constraint{$k≥0$} ⏎
        $◇\produce∘~|~+~|~-$ ⏎
        $\vv{τ}\produceτ₁,…,τₖ$ &
          \constraint{$k≥0$} ⏎
        $τ\produce γ(\vv{τ})~|~x$ ⏎
        $q \produce t_⊥\subtype t_⊤$ ⏎
        $t\produceγ(\vv{t})$ ⏎
        $\vv{t} \produce t₁,…,tₖ$ &
          \constraint{$k≥0$} ⏎
      \end{tabular} &
      \renewcommand{\arraystretch}{1.1}
      \begin{tabular}{ll}
        $P$ & program ⏎
        $δ$ & class declaration ⏎
        $Δ$ & class table; a set of class declarations ⏎
        $γ$ & class name drawn from finite set~$Γ$ ⏎
        $τ$ & unground type ⏎
        $\vv{τ}$ & sequence of unground types ⏎
        $x$ & type parameter drawn from finite set~$X$ disjoint from~$Γ$ ⏎
        $◇x$ & type parameter with variance ⏎
        $\vv{◇x}$ & sequence of type parameters with variances ⏎
        $∘$ & invariance ⏎
        $+$ & covariance ⏎
        $-$ & contravariance ⏎
        $q$ & subtyping query ⏎
        $t$ & ground type ⏎
        $\vv{t}$ & sequence of ground types ⏎
      \end{tabular} ⏎
      \hline
      (a) \textbf{Syntax} & (b) \textbf{Nomenclature} ⏎
      \hline\hline
      \renewcommand{\arraystretch}{3}
      \begin{tabular}{c}
        $\infer{t \inheritsτ[s]}{t⊑ₛ γ(\vv{x}) & γ(\vv{x}):τ}$ ⏎
      \end{tabular} &
      \renewcommand{\arraystretch}{3}
      \begin{tabular}{c}
        $\rulename{Var}\infer{γ(\vv{t}) \subtype γ(\vv{t}')}{∀ i, tᵢ \subtype_{\vof[i]{γ}} t'ᵢ}$ \quad
        $\infer{t \subtype₊ t'}{t \subtype t'}$ \quad
        $\infer{t \subtype_∘ t}{}$ \quad
        $\infer{t \subtype₋ t'}{t \suptype t'}$ ⏎
        $\rulename{Super}\infer{t \subtype t”}{t \inherits t' & t' \subtype t”}$
      \end{tabular}
 ⏎
      \hline
      (c) \textbf{Inheritance} & (d) \textbf{Subtyping} ⏎
      \hline\hline
      \multicolumn{2}{|c|}{
        \renewcommand{\arraystretch}{2}
        \begin{tabular}{c}
          \rulenames{Acyclic ⏎ Inheritance}
          $γ(\vv{τ}) \inherits⁺ γ'(\vv{τ}')~⇒~γ≠γ'$ ⏎
          \raisebox{2ex}{
            \rulenames{Valid ⏎ Variance}
            $¬◇=\begin{cases}
- &◇=+⏎
                  ∘ &◇=∘ ⏎
            +&◇=-
            \end{cases}$} \quad
          $\infer{\vv{◇ x}⊢xᵢ}{◇ᵢ∈❴∘,+❵}$ \quad
          $\infer{\vv{◇ x}⊢γ(\vv{t})}{∀ i, \begin{cases}
            \vof[i]{γ}∈❴∘,+❵⇒\vv{◇ x}⊢τᵢ ⏎
            \vof[i]{γ}∈❴∘,- ❵⇒¬\vv{◇ x}⊢τᵢ
            \end{cases}}$ \quad
          $\infer{γ(\vv{◇ x}):τ~\checkmark}{\vv{◇ x}⊢τ}$
        \end{tabular}
      } ⏎
      \hline
      \multicolumn{2}{|c|}{(e) \textbf{Well-Formedness}} ⏎
      \hline
    \end{tabular}
  \end{adjustbox}
  \caption{Kennedy and Pierce's type system (without mixin inheritance)}
  \label{figure:type:system}
\end{figure}

The full syntax of the type system is given in
\aref[(a)]{figure:type:system} and a nomenclature is available in
\aref[(b)]{figure:type:system}.
A program consists of a class table~$Δ$ and a subtyping query~$q$.
A class table over a finite set of class names~$Γ$ is a finite set of
inheritance rules $\boldsymbol\delta$. An inheritance rule~$\delta=γ(\vv{◇ x}):\vv{τ}$ declares a
parameterized class (polytype) named~$γ$ and its inheritance from (super-)
types~$\vv{τ}$; to highlight that a sequence of supertypes~$\vv{τ}$ is empty,
we write~$∅$ after the colon. We require supertypes to reference only class
parameters,~$\params(\vv{τ})⊆\vv{x}$. Although Kennedy and Pierce included
mixin inheritance in their type system, allowing a class to inherit one of
its parameters, we chose to exclude it; accordingly,~$\vv{τ}∩\vv{x}=∅$ (nevertheless, the implications of
mixin inheritance are discussed in \cref{section:discussion}). Variance is denoted by
a diamond~$◇$, which can be covariant,~$◇=+$, contravariant,~$◇=-$, or
invariant,~$◇=∘$. The variance of the \nth{i} parameter of class~$γ$ is
denoted~$\vof[i]{γ}$, and the vector describing all~$γ$ parameter variances is
denoted~$\vof{γ}$. A subtyping query~$t_⊥\subtype t_⊤$ (also
written~$t_⊥\subtype_Δ t_⊤$), for subtype~$t_⊥$ and supertype~$t_⊤$, both
ground, type checks if and only if~$t_⊥$ is a subtype of~$t_⊤$ (considering
class table~$Δ$).

\paragraph*{Example}
For example, we can write the \CSharp program of \cref{lst:example:cff:program}
using abstract syntax as follows (recall that we use an abbreviated notation of monadic terms):
\begin{equation}\label{eq:example:program}
\begin{aligned}
\Delta&\begin{cases}
	&a(+x): \emptyset \\
	&b(+x): \emptyset \\
	&E: \emptyset \\
	&v_0(\circ x): a v_0 a x, a a x, a x, b v_0 b x, b b x, b x \\
\end{cases} \\
q&\begin{cases}
	&\:v_0 E \subtype a b b a b b a E
\end{cases}
\end{aligned}
\end{equation}
Class table $\Delta$ has four classes, $a$, $b$, $E$, and $v_0$.
In \CSharp, type parameters are invariant by default, covariance is denoted by
the keyword \kk{out}, and contravariance by the keyword \kk{in}.
Therefore, the type parameters of $a$ and $b$ are covariant ($\vof[1]{a}=\vof[1]{b}=+$),
and the type parameter of $v_0$ is invariant ($\vof[1]{v_0}=\circ$).
Query $q$ is the subtyping query invoked by the variable assignment in line 3 of \cref{lst:example:cff:program}.

According to the type inheritance relation $t \inherits t'$, defined in \aref[(c)]{figure:type:system},
type $t$ inherits type $t'$ if $t'$ matches one of the supertypes of $t$'s root.
Unlike Kennedy and Pierce, we distinguish between the type inheritance relation and
class inheritance,~$γ:γ'$, found in class declarations.
We also use the transitive closures of these:
\begin{equation}\label{eq:transitive:inheritance}
\begin{aligned}
  \infer{γ(\vv{x}):^* γ(\vv{x})}{} \qquad
  \infer{γ(\vv{x}):^*τ[\vv{x}'←\vv{τ}]}
    {γ(\vv{x}): γ'(\vv{τ}) & γ'(\vv{x}'):^*τ}
\end{aligned}
\end{equation}
The same rules apply for transitive type inheritance. The asterisk is replaced
with a plus~$+$ in the non-reflexive variants~$:⁺$ and~$\inherits⁺$.

The rules of subtyping are detailed in \aref[(d)]{figure:type:system}.
The \textsc{Super} rule establishes that a class is a subtype of its supertypes.
The \textsc{Var} rule handles the type parameters, depending on their variance:
Statement~$γ(t) \subtypeγ(t')$ implies~$t \subtype t'$ if~$γ$'s type parameter
is covariant,~$t=t'$ if invariant, and~$t' \subtype t$ if contravariant.
If~$γ$ accepts multiple type parameters, then each may have its own variance.
We also define transitive subtyping:
\begin{equation}
\begin{aligned}
  \infer{t \subtype^* t}{\relax} \qquad
  \infer{t \subtype^* t”}{t \subtype t' & t' \subtype^* t”}
\end{aligned}
\end{equation}
For example, we can use the \textsc{Super} and \textsc{Var} rules to prove the
subtyping query of \cref{eq:example:program}:
\begin{equation*}
\newcommand{\alert}[1]{#1}
\begin{tabular}{rll}
	& $\alert{v_0} E \subtype a b b a b b a E$ & \\
	$\Leftarrow$\quad & $\alert{a} v_0 a E \subtype \alert{a} b b a b b a E$ & $(\textsc{Super},~v_0(\circ x) : a v_0 a x)$ \\
	$\Leftarrow$\quad & $\alert{v_0} a E \subtype b b a b b a E$ & $(\textsc{Var},~\vof[1]{a}=+)$ \\
	$\Leftarrow$\quad & $\alert{b} v_0 b a E \subtype \alert{b} b a b b a E$ & $(\textsc{Super},~v_0(\circ x) : b v_0 b x)$ \\
	$\Leftarrow$\quad & $\alert{v_0} b a E \subtype b a b b a E$ & $(\textsc{Var},~\vof[1]{b}=+)$ \\
	$\Leftarrow$\quad & $\alert v_0 b b a E \subtype \alert{b} a b b a E$ & $(\textsc{Super},~v_0(\circ x) : b v_0 b x)$ \\
	$\Leftarrow$\quad & $\alert{v_0} b b a E \subtype a b b a E$ & $(\textsc{Var},~\vof[1]{b}=+)$ \\
	$\Leftarrow$\quad & $a b b a E \subtype a b b a E$ \checkmark & $(\textsc{Super},~v_0(\circ x) : a x)$
\end{tabular}
\end{equation*}
Note that this subtyping query is recursive because the type parameters of classes $a$ and $b$
are covariant. If the parameters were contravariant, then the query would reverse ($\suptype$)
after applying the \textsc{Var} rule.

A class table is well-formed if it complies with the rules in
\aref[(e)]{figure:type:system}. First, inheritance must be acyclic, i.e., a
class may not (transitively) inherit itself. Thus, the class inheritance closure~$:^*$ is
finite. Kennedy and Pierce also explain that transitive subtyping and semantic
soundness rely on certain restrictions to class inheritance,~$γ(\vv{◇ x}):τ$.

To demonstrate the problem, consider the following \emph{invalid} class declarations:
\begin{equation*}
\begin{aligned}
	& a(-x) : \emptyset \\
	& b(+y) : a(y)
\end{aligned}
\end{equation*}
Type parameter $y$ of class $b$ is covariant, but it also substitutes for $a$'s
contravariant parameter~$x$.
This means that $y$ can be used both covariantly and contravariantly.
For instance, if $c \inherits d$, then
\begin{equation*}
	b(c) \subtype b(d)
\end{equation*}
but also
\begin{equation*}
	b(d) \subtype a(c)
\end{equation*}
One way to solve this issue is by changing $b$'s supertype to $a(a(y))$:
\begin{equation*}
	b(+y) : a(a(y))
\end{equation*}
Once done, $y$ behaves covariantly, as expected:
\begin{equation*}
\begin{tabular}{rll}
	& $bc \subtype aad$ & \\
	$\Leftarrow$\quad & $aac \subtype aad$ & $(\textsc{Super},~b(+y) : aay)$ \\
	$\Leftarrow$\quad & $ac \suptype ad$ & $(\textsc{Var},~\vof[1]{a}=-)$ \\
	$\Leftarrow$\quad & $c \subtype d$ \checkmark & $(\textsc{Var},~\vof[1]{a}=-)$ \\
\end{tabular}
\end{equation*}

Intuitively, one may think of a covariant type parameter as defining a
‟positive position”, and a contravariant one a ‟negative position”. A covariant
type parameter~$x∈\vv{x}$ can be placed in supertype~$τ$ only in positive
positions, and if it is contravariant, only in negative positions. A
sub-tree~$τ'$ of~$τ$ in a negative position reverses its polarity: Positive
positions become negative, and negative positions become positive. The initial
position of~$τ$ is positive. In this context, an invariant position is both
positive and negative, and so are invariant type parameters. Thus, invariant
type parameters can appear in any position, while a type in an invariant
position contains only invariant type parameters, because its positions are both
positive and negative.

For example, consider the following class table:
\begin{equation}\label{eq:example:class:table}
  \begin{aligned}
    a(+x):∅&\qquad\qquad d(+x₁): adax₁ ⏎
    b(-x):∅&\qquad\qquad e(-x₂): bbabx₂ ⏎
    c(∘ x):∅&\qquad\qquad f(-x₃,∘x₄,+x₅): af(x₄, ax₄, x₄)
  \end{aligned}
\end{equation}
This class table, ranging over classes~$a,b,c,d,e,$ and $f$, is well-formed.
The table has no inheritance cycles, e.g., although type~$d$ appears in the
type it inherits, we consider only the roots of supertypes in checks of
circularity.  Restrictions on the use of type parameters are also respected.
Covariant type parameter~$x₁$ in type~$adax₁$ is in a positive position, as both~$a$ and~$d$ are covariant.
Contravariant type parameter~$x₂$ in type~$bbabx₂$ is in a negative position: Type~$b$ reverses the
polarity an odd number of times and type~$a$ does not change it, so the final position is negative.
Invariant type parameter~$x₄$ can be placed in all the positions of~$f$ in type~$af(x₄, ax₄, x₄)$,
and it is the only parameter suitable for the invariant (second) position.
The~\CSharp programming language enforces the rules of \aref[(e)]{figure:type:system};
in \cref{lst:example:class:table}, we encode the class table of \cref{eq:example:class:table} as
\CSharp interface declarations that compile as expected.
\begin{code}[float=ht,style=csharp,caption={The class table of \cref{eq:example:class:table} encoded in~\CSharp},label={lst:example:class:table}]
interface a<out x> {}        interface d<out x1> : a<d<a<x1>>> {}
interface b<in x> {}         interface e<in x2> : b<b<a<b<x2>>>> {}
interface c<x> {}            interface f<in x3, x4, out x5> : a<f<x4,a<x4>,x4>> {}
\end{code}

Let the type system featured in \cref{figure:type:system} be denoted as \TKP.
\citet{Kennedy:Pierce:07} proved that \TKP is undecidable.
Nevertheless, they introduced three restrictions on \TKP, each defining a
decidable type system fragment of \TKP.
Each restriction removes a certain feature of the type system:
contravariance (C), expansive inheritance (X), or multiple instantiation inheritance (M).
We denote a type system by the set of features used by its programs, e.g.,
\Tcxm is an alias of \TKP, and its subset containing only programs without
contravariance is denoted by \Txm.
Each element of the power set of~$❴\text{C}, \text{X}, \text{M}❵$ defines
a type system. We organize the eight type systems in a lattice, depicted in
\cref{figure:lattice}, ordered with respect to inclusion (of features
and programs). The bottom of the lattice is \Tbot, bare of any special features.
All these type systems employ polyadic parametric polymorphism, i.e.,
where a class can have multiple type parameters.

The first restriction of Kennedy and Pierce forbids contravariant type
parameters,~$-x$, yielding type system \Txm.  The second restriction disallows
multiple instantiation inheritance, making inheritance deterministic:
If~$γ(\vv{x}):⁺ γ'(\vv{τ})$ and~$γ(\vv{x}):⁺ γ'(\vv{τ}')$,
then necessarily~$\vv{τ}=\vv{τ}'$.
For instance, the class table of \cref{eq:example:program} does not conform
to this restriction, as class $v_0$ extends both $a(x)$ and $a(a x)$.
Kennedy and Pierce suspected that single instantiation inheritance, featured in \Tcx, is not enough
to achieve decidability.
\citet{Grigore:2017} verified the suspicion by simulating Turing machines
with subtyping queries in \Java, which employs deterministic inheritance.
To make single inheritance decidable, it had to be further restricted; the resulting type
system will not be discussed in this paper.

The third restriction, due to \citet{Viroli:00}, limits the growth of types during subtyping.
Consider, for example, the class table depicted in \cref{eq:example:class:table}.
Given type $d t$, for some type argument $t$, on either side of a subtyping query,
it can change to $adat$ by applying the \textsc{Super} rule and then to $dat$ by \textsc{Var},
yielding the type $dt'$ for $t'=at$.
Type $d t$ can, therefore, expand indefinitely during subtyping.
Kennedy and Pierce named this kind of inheritance \emph{recursively expansive} and showed its
connection to infinite subtyping proofs.

Originally, Viroli illustrated this growth of types by a graph, depicting the dependencies between
the different type variables employed by the class table.
In this paper, however, we use a more approachable definition, following an equivalence result of
Kennedy and Pierce.
Consider the following definition:
\begin{definition}[Inheritance and decomposition closure]\label{definition:inheritance:decomposition:closure}
  Set~$T$ of types is closed under inheritance and decomposition if for every type~$γ(\vv{t})$ in~$T$,
  its sub-trees are in~$T$,~$\vv{t}⊂T$, and so are its supertypes,~$γ(\vv{t}) \inherits t⇒ t∈T$.
  The minimal super-set of set~$S$ closed under inheritance and decomposition is denoted~$\cl(S)$.
  (The curly braces of set~$S$ may be omitted inside~$\cl(·)$, e.g.,~$\cl(t₁, t₂)$.)
\end{definition}
With the class table of \cref{eq:example:class:table}, for instance,~$\cl(dt)⊃❴daⁿt~|~n≥0❵$,
as explained above.
\citet{Kennedy:Pierce:07} noted that the set of types involved in a subtyping query~$t_⊥\subtype t_⊤$ is
contained in their inheritance and decomposition closure,~$\cl(t_⊥, t_⊤)$, following the
\textsc{Super} (inheritance) and \textsc{Var} (decomposition) rules.
If a closure~$\cl(t_⊥, t_⊤)$ is finite, then the subtyping query~$t_⊥\subtype t_⊤$ is decidable,
as we can trace infinite subtyping cycles.
Therefore, if the closure~$\cl(t, t')$ of any two types~$t$ and~$t'$ is finite,
subtyping is decidable, and the class table is called finitary.
Kennedy and Pierce proved that subtyping with non-expansive inheritance is decidable,
by showing that a class table is non-expansive if and only if it is finitary.
Thus, we can use \cref{definition:inheritance:decomposition:closure} to define
non-expansive inheritance directly:
\begin{definition}[Non-expansive class table]\label{definition:non:expansive}
	Class table $\Delta$ is non-expansive if and only if the inheritance and
	decomposition closure of any set of types~$T$ is finite, $|\cl(T)|<\infty$.
\end{definition}
\noindent
We denote the type system fragment where all the class tables are non-expansive as \Tcm.

If a class table~$Δ$ conforms to type system~$𝓢$, we write~$𝓢⊢Δ$.  For example,
class table~$Δ$ of \cref{eq:example:class:table} conforms to \TKP, written~$\TKP⊢Δ$, as
it is well-formed, but not to \Tcm, written~$\Tcm \nvdashΔ$, as its inheritance is
expansively recursive.

	\section{Subtyping Machines as Forest Recognizers}
	\label{section:machines}
	A subtyping machine is a set of type system definitions that utilize subtyping and
its rules to perform a computation. \citet{Grigore:2017}, who coined the term,
used subtyping machines to simulate Turing machines in \Java, thus proving its
type system is undecidable. 

This section formalizes subtyping machines as
general computation devices, following a standard approach in fluent API
research~\cite{Gil:19,Yamazaki:2019,Gil:20}.
Subtyping machines recognize formal forests, i.e., sets of trees
or terms. In \cref{section:expressiveness}, we use this formalism to classify
the computational power of decidable subtyping in its various forms.

This section's running example is the program shown in \cref{lst:example:cff:program}.
We focus on the type declaration in the first two lines of the code:
\begin{code}
interface a<out x> {} interface b<out x> {} interface E {}
class v0<x> : a<v0<a<x>>>, a<a<x>>, a<x>, b<v0<b<x>>>, b<b<x>>, b<x> {}
\end{code}
Let us name this program \textsf{Pali}.
In \cref{section:aa}, we introduced \textsf{Pali} as an example of a subtyping
machine that recognizes palindromes, but did not explain what ``recognize'' means.
Our goal is, therefore, to formalize an interpretation of \textsf{Pali} as
a formal forest of palindromes.

In our type system, presented in \cref{figure:type:system}, a program~$P=Δ~q$
comprises a set of declarations~$Δ$ constituting a class table and a query (or
assertion)~$q=t_⊥\subtype t_⊤$. The query succeeds (compiles) if and only if
type~$t_⊥$ is a subtype of type~$t_⊤$ considering class table~$Δ$ and the
subtyping semantics of the type system,~$t_⊥\subtype_Δ t_⊤$. Therefore, a
given class table defines the set of queries that compile against it; we call
this set the \emph{extent} of the class table.

The extent of \textsf{Pali}'s class table $\Delta$ (written using abstract
syntax in \cref{eq:example:program}) contains, for instance, the queries
\begin{equation*}
	\text{$v_0 E \subtype aE$, \qquad $v_0 E \subtype abaE$, \qquad and\ \ \ $v_0 E \subtype abbabbaE$,}
\end{equation*}
because they type check against $\Delta$.
In general, the extent of $\Delta$ contains all the queries of the form
\begin{equation*}
	v_0 E \subtype p E
\end{equation*}
where $p$ is a palindrome over $\{a, b\}$ (we prove this in \cref{section:non:contravariant}).
$\Delta$'s extent, however, also contains queries that do not include palindromes at all,
such as
\begin{equation*}
	\text{$E \subtype E$, \qquad $v_0 a E \subtype baE$, \qquad and\ \ \ $a v_0 E \subtype abE$.}
\end{equation*}
To correlate \textsf{Pali}'s extent with a forest of palindromes,
we first need to restrict the contents of its queries.

Recall that formal grammars distinguish between
terminal and non-terminal (variable) symbols, drawn from sets~$Σ$ and~$V$,
respectively. This dichotomy allows grammars to use variables as auxiliary
symbols, which cannot appear as part of derived words. To make our definition of
subtyping machines effective, we introduce an analogous separation in~$Γ$, the
set of class names in~$Δ$. Let sub-alphabet~$Σ_⊥⊆Γ$ and
super-alphabet~$Σ_⊤⊆Γ$ be the domains of subtypes~$t_⊥$ and
supertypes~$t_⊤$ in subtyping queries, i.e., types~$t_⊥$
and~$t_⊤$ consist of type names drawn from designated~$Γ$ subsets~$Σ_⊥$
and~$Σ_⊤$.
We amend Kennedy and Pierce's type system, defined in \cref{figure:type:system},
to include the sub- and super-alphabets:
\begin{equation}\label{eq:amendment}
	q \produce t_⊥\subtype t_⊤ \quad {\scriptstyle\constraint{t_⊥∈Σ_⊥^▵,~t_⊤∈Σ_⊤^▵}} \qquad (Σ_⊥^▵,Σ_⊤^▵\subseteq\Gamma)
\end{equation}
From this point on, we discuss programs that include the amendment presented in \cref{eq:amendment}.
We define the extent of class table~$Δ$ as the subset of~$Σ_⊥^▵⨉Σ_⊤^▵$:
\begin{definition}[The extent of a class table]
  \label{definition:class:table:extent}
  Let~$Δ$ be a class table, as described in \cref{figure:type:system}.
  The extent of~$Δ$, denoted~$L(Δ)$, contains all the pairs of types~$⟨t_⊥, t_⊤⟩$ such
  that the subtyping query~$t_⊥\subtype t_⊤$ type checks against~$Δ$: \[
    L(Δ)=❴⟨t_⊥,t_⊤⟩~|~t_⊥∈Σ_⊥^▵,~t_⊤∈Σ_⊤^▵,~t_⊥\subtype_Δ t_⊤❵
    \qquadΣ_⊥,Σ_⊤⊆Γ
\]
\end{definition}

\cref{definition:class:table:extent} gives us more control over the expressiveness
of programs.
The extent of a class table, however, does not yet describe a formal forest.
Although types are homomorphic to trees, an extent is a set of type pairs
$⟨t_⊥,t_⊤⟩$, rather than a set of lone types.
To solve this issue, we fix the subtype~$t_⊥$ to some constant type that is
no longer a variable part of the subtyping query.
Then, a class table~$Δ$ defines a (tree) language of~$t_⊥$ supertypes;
this set is the projection of type~$t_⊥$ onto extent~$L(Δ)$:
\begin{definition}[The language of a class table]
  \label{definition:class:table:language}
  Let~$Δ$ be a class table and~$t_⊥$ be a type, as described in \cref{figure:type:system}.
  The language of~$Δ$ over~$t_⊥$, denoted~$L_{t_⊥}(Δ)$, contains all the supertypes of~$t_⊥$ against~$Δ$, \[
    L_{t_⊥}(Δ)=❴t_⊤~|~⟨t_⊥, t_⊤⟩∈L(Δ) ❵=❴t_⊤~|~t_⊤∈{Σ_⊤}^▵,~t_⊥\subtype_Δ t_⊤❵
    \qquadΣ_⊤⊆Γ
\] \end{definition}

We can now use \cref{definition:class:table:language} to describe \textsf{Pali}'s
language over subtype $t_\bot = v_0 E$ as a forest of palindromes:
\begin{equation}\label{eq:language:pali}
	L(\textsf{Pali})_{v_0E} = \{ p E ~|~ \text{$p$ is a palindrome} \} \qquad (\Sigma_\top = \{a, b, E\})
\end{equation}
In other words, \textsf{Pali} is a subtyping machine that recognizes palindromes
as supertypes of type $t_\bot = v_0 E$.
We prove \cref{eq:language:pali} in \cref{section:non:contravariant}.

Considering \cref{definition:class:table:language}, we see that class tables describe formal
languages (forests), as grammars do, and their
compilers serve as language (forest) recognizers, similarly to grammar parsers.
Correspondingly, a type system is a family of programs that recognizes a class of languages:
\begin{definition}[The class of languages of a type system]
  \label{definition:type:system:extent}
  The computational class (expressiveness) of subtyping-based type system~$𝓢$, denoted~$𝓛(𝓢)$,
  is the set of languages described by its programs: \[
    𝓛(𝓢)=❴L_{t_⊥}(Δ)~|~t_⊥∈Σ_⊥^▵,~𝓢⊢Δ❵ \qquad Σ_⊥⊆Γ
\] 
\end{definition}
\noindent (Recall that~$𝓢⊢Δ$ means class table~$Δ$ is well-formed
and conforms to the restrictions of type system~$𝓢$.)

	\section{Expressiveness of Decidable Subtyping}
	\label{section:expressiveness}
	Recall that the type system of Kennedy and Pierce \TKP features contravariance (C),
expansive inheritance (X), and multiple instantiation inheritance (M).
Subset~$S$ of~$❴C,X,M❵$ describes the type system fragment of
\TKP endorsing only the features included in~$S$.
Subtyping in \TKP and \Tcx is known to be undecidable, but it is decidable in
the other six fragments~\cite{Kennedy:Pierce:07,Grigore:2017}.
In \cref{section:machines}, we described class tables as subtyping machines that recognize formal forests
(\cref{definition:class:table:language}) and defined the expressiveness
of a type system as the class of forests that its tables recognize (\cref{definition:type:system:extent}).
Now we use these formalizations to describe the computational power of Kennedy and Pierce's
decidable type system fragments.

Given type system~$𝓢$, we show~$𝓢$ is at least as
complex as class of languages~$𝓛$ by proving that any language~$ℓ$
in~$𝓛$ is described by some class table~$Δ$ and subtype~$t_⊥$ of~$𝓢$: \[
  \big(∀ℓ∈𝓛,~∃Δ, t_⊥,~𝓢⊢Δ,~ℓ=L_{t_⊥}(Δ)\big)~⇒~𝓛⊆𝓛(𝓢)
\] To show that type system~$𝓢$ is not more complex than class of
languages~$𝓛$, we prove that any class table~$Δ$ and subtype~$t_⊥$ describe an~$𝓛$
language: \[
  \big(∀Δ, t_⊥,~𝓢⊢Δ,~∃ℓ∈𝓛,~ℓ=L_{t_⊥}(Δ)\big)~⇒~𝓛(𝓢)⊆𝓛
\] If the expressiveness of~$𝓢$ is bounded by~$𝓛$ from both directions,
then it is exactly~$𝓛$: \[
  𝓛⊆𝓛(𝓢)⊆𝓛~⇒~𝓛(𝓢)=𝓛
\] In all the following proofs, formal languages are specified by their corresponding grammars.

\subsection{Non-Expansive Subtyping Is Regular}

\citet{Greenman:14} raised a concern about non-expansive inheritance restricting
some conventional class table use-cases.
We go one step further by showing that the expressiveness of non-expansive class tables is regular,
i.e., of the lowest tier.

\begin{theorem}\label{theorem:regf:lower}
  Any regular forest can be encoded in a class table using only multiple instantiation inheritance: \[
    \|NRF|⊆𝓛(\Tm)
\] \end{theorem}

\paragraph*{Intuition}
We want to simulate the derivation of a regular tree grammar with subtyping.
In the first step, we encode every terminal node $\sigma$ of rank $k$ by
class $\sigma$ that has $k$ type parameters.
Thus, terminal trees (derived from the grammar) and types (in subtyping queries)
coincide, as both use the same term notation and base alphabet, e.g.,
$\sigma_1(\sigma_2, \sigma_3 \sigma_4)$ is both a tree and a type.
Next, every regular derivation rule $
	v \produce \sigma(\vv{v})
$ is directly encoded by a corresponding inheritance rule $
	v : \sigma(\vv{v}).
$
With these definitions, the subtyping query $v \subtype t$ simulates the grammar
derivation $v \rightarrow^* t$:
Let $t=\sigma(\vv{t})$.
Grammar variable $v$ derives $t$ if $v$ derives $\sigma$, $v \produce \sigma(\vv{v})$
and then, recursively, variables $\vv{v}$ derive $\sigma$'s children, $\vv{v} \rightarrow^* \vv{t}$.
Correspondingly, the subtyping query $v \subtype t$ type checks if class $v$
inherits type $\sigma$, $v : \sigma(\vv{v})$ and then, recursively, types $\vv{v}$ are subtypes
of $\vv{t}$, $\vv{v} \subtype \vv{t}$.
The recursion of subtyping originates from the \textsc{Var} rule, under the condition that
the type parameters of $\sigma$ are covariant.

\paragraph*{Example}
Consider the regular tree grammar in \cref{eq:example:regular:grammar} in \cref{section:forests},
describing lists of Peano numbers.
This grammar can be encoded in a~\CSharp program with neither contravariance nor expansive recursive inheritance, as
demonstrated in \cref{lst:example:regular:program}.
\begin{code}[style=csharp,caption={\protect\CSharp program encoding the grammar of \protect\cref{eq:example:regular:grammar}},label={lst:example:regular:program}]
interface z {}                    // Terminal leaf $\color{comment}\+z+$
interface s<out x> {}             // Terminal node $\color{comment}\+s+(x)$
interface nil {}                  // Terminal leaf $\color{comment}\+nil+$
interface cons<out x1, out x2> {} // Terminal node $\color{comment}\+cons+(x1, x2)$
interface Nat: // Variable $\color{comment}\-Nat-$
    z,         // Production $\color{comment}\-Nat-\produce\+z+$
    s<Nat> {}  // Production $\color{comment}\-Nat-\produce\+s+(\-Nat-)$
interface List:        // Variable $\color{comment}\-List-$
    nil,               // Production $\color{comment}\-List-\produce\+nil+$
    cons<Nat, List> {} // Production $\color{comment}\-List-\produce\+cons+(\-Nat-, \-List-)$
// Subtyping query $\color{comment}\-List-\subtype\+cons+(\+ssz+, \+cons+(\+sz+, \+cons+(\+z+, \+nil+)))$
cons<s<s<z>>, cons<s<z>, cons<z, nil>>> t=(List) null;
\end{code}
The grammar terminals \+s+, \+z+, \+cons+, and \+nil+ are encoded
by interfaces with covariant type parameters, and the grammar variables
\-Nat- and \-List- are encoded by interfaces whose supertypes are the possible productions
of each variable.
The last line of the listing contains a variable assignment that
ensures type~$v₀=\-List-$ is a subtype of
the~\CSharp type encoding the list~$⟨{}2,1,0⟩$ (described in \cref{eq:example:list}).
This assignment compiles, as expected.

\begin{proof}[Proof of \cref{theorem:regf:lower}]
  A given regular forest is described by a regular tree grammar
  $G=⟨Σ, V, v₀, R⟩$.
  We encode~$G$ as class table~$Δ$,~$\Tm⊢Δ$, such that
  the language of~$Δ$ (over~$t_⊥=v₀$) is exactly the language described
  by~$G$,~$L(G)$.

  \proofpart{Construction}
  The class table~$Δ$ employs type names drawn from~$Γ=Σ∪V$.
  The number of type parameters of each type~$σ∈Σ$ is identical to its rank in~$G$,
  and all type parameters are covariant (nodes~$v∈V$ are leaves).
  For each regular production rule~$v \produceσ(\vv v)∈R$, set
  type~$σ(\vv v)$ as a super-type of type~$v$,~$v:σ(\vv v)$.
  We set the subtype to~$t_⊥=v₀$, and the super-alphabet to~$Σ_⊤=Σ$.

  \proofpart{Correctness of the construction}
  Subtyping against~$Δ$ directly emulates derivations of grammar~$G$,
  where the left-hand side is the current variable (tree form),
  starting with the initial variable~$v₀$, and the right-hand side is the
  derived tree.
  We show that a tree~$t$ is produced by variable~$v$ if and only if type~$t$ is a supertype
  of type~$v$,
  \begin{equation}\label{eq:statement:1}
    v→^*_G t~⇔~v \subtype t
  \end{equation}
  by induction on the height of~$t$,~$n=\height(t)$.

  If~$n=0$, then tree~$t$ is a leaf,~$t=σ$,~$σ∈Σ$.
  Variable~$v$ derives~$σ$ only if production~$v \produceσ$ is in~$R$.
  In that case, type~$σ$ is a supertype of~$v$,~$v:σ∈Δ$,
  so~$v$ is a subtype of~$t$ by inheritance: \[
    v→^*_Gσ~⇔~v \produceσ∈R~⇔~
    v:σ∈Δ~⇔~v \subtype \sigma
\] Suppose \cref{eq:statement:1} holds for~$n≥0$,
  and let~$t$ be a tree of height~$n+1$,~$t=σ(\vv t)$,~$σ∈Σ$,~$\vv t∈Σ^▵$,
  where the heights of sub-trees~$\vv t$ are~$n$ at most.
  If variable~$v$ derives~$t$, then there is a production~$v \produceσ(\vv v)$ in~$R$, for
  which every child variable~$vᵢ∈\vv v$ derives the corresponding child tree~$tᵢ$,~$\vv v→^*_G \vv t$.
  This production ensures type~$v$ inherits type~$σ(\vv v)$.
  As the heights of~$\vv t$ are no higher than~$n$, we can use the inductive assumption to deduce~$\vv v \subtype \vv t$.
  Because the type parameters of~$σ$ are covariant, we can show that~$v$ is a subtype of~$t$ by
  applying inheritance, decomposition, and the inductive assumption: \[
    \begin{aligned}
      v \subtypeσ(\vv t)~&⇐~v:σ(\vv v)~∧~σ(\vv v) \subtypeσ(\vv t) ⏎
σ(\vv v) \subtypeσ(\vv t)~&⇐~\vv v \subtype \vv t~∧~\vof{σ}=\vv+
    \end{aligned}
\] Overall, we have now proved \cref{eq:statement:1}: \[
    v→^*_G t~⇔~v \produceσ(\vv v)∧\vv v→^*_G \vv t~⇔~
    v:σ(\vv v)∧\vv v \subtype \vv t~⇔~v \subtype t
\] By applying \cref{eq:statement:1} to type~$t_⊥=v₀$, we get \[
    v₀→_G^* t~⇔~t_⊥\subtype t,
\] i.e.,~$L(G)=L_{t_⊥}(Δ)$ by definition.

  Class table~$Δ$ conforms to the restrictions of the type system,~$\Tm⊢Δ$.
  First, the class table is well-formed. Only types~$v∈V$ inherit types of the form~$σ(\vv v)$
  for~$σ∈Σ$, so there are no inheritance cycles and no type parameters in the wrong position.
  Second,~$Δ$ does not employ contravariance and is non-expansive: 
  The structure of inheritance rules guarantees that any type in~$\cl(t)$ is
  higher than~$t$ by one level at the most.
\end{proof}

\begin{theorem}\label{theorem:regf:upper}
  Non-expansive class tables describe regular forests: \[
𝓛(\Tcm)⊆\|NRF|
\] \end{theorem}

\paragraph*{Intuition}
A grammar variable $v \in V$ encodes a set of queries $Q$, $v=Q$.
These queries restrict the set of types (trees) that $v$ can derive.
Each query $q \in Q$ is of the form $t \sim x$, where
\1 $t$ is a type,
\2 relation $\sim$ is subtyping ($\sim=\subtype$), suptyping (the opposite
	of subtyping, $\sim=\suptype$), or equivalence ($\sim$ is $=$), and
\3 symbol $x$ is a placeholder for the type to be derived.
If query $q=t \sim x$ is included in variable $v$, then
our construction ensures that the type derived by $v$ satisfies $q$: \[
	t \sim x\in v,~v \rightarrow^* t' ~\Rightarrow~ t \sim t'
\]
For example, the initial variable $v_0$ includes a single query $t_\bot \subtype x$,
so each type $t'$ derived by $v_0$ satisfies $t_\bot \subtype t'$---the set of
these types is, by \cref{definition:class:table:language}, exactly the language of
the class table over $t_\bot$.

As the class table is non-expansive, both sets of variables $V$ and productions
$R$ are finite and computable.
If the initial query $t_\bot \subtype x$ is reduced to some
other query $t \sim x$, then, as observed by \citet{Kennedy:Pierce:07},
type $t$ is included in $t_\bot$'s inheritance and decomposition closure
$\cl(t_\bot)$ (see \cref{definition:inheritance:decomposition:closure}).
Because the class table is non-expansive, this set must be finite
(see \cref{definition:non:expansive}).
Thus, for a general query $t \sim x$, there are $|\cl(t_\bot)|$ choices
for $t$, three choices for $\sim$, and, consequently, $3|\cl(t_\bot)|$ possible
queries overall.
As variables are sets of queries, we get that the
set of variables $V$ is finite, $|V|\le 2^{3|\cl(t_\bot)|}$.
A variable $v$ can derive terminal $\gamma$, $v \produce \gamma(\vv{v})$,
if and only if the queries in $v$ can be reduced to the queries in variables
$\vv{v}$; variables $\vv{v}$ then recursively simulate the subtyping algorithm with $\gamma$'s children.
In the proof below we show that as the class table is non-expansive,
it is possible to compute the derived queries $\vv{v}$ in finite time.

\paragraph*{Example}
Consider the following non-expansive class table~$Δ$, first presented in \citet{Kennedy:Pierce:07}: \[
\begin{aligned}
	N(-x) &:∅⏎
	C &: NNC
\end{aligned}
\]
Let subtype~$t_⊥=C$ and super-alphabet~$Σ_⊤=❴N, C ❵$;
then, class table $\Delta$ is encoded by the following grammar $G$:
\begin{equation*}
	\begin{array}{ll}
		v₀=❴C \subtype x ❵ &\quad v₀→C ⏎
		v₁=❴NC \suptype x ❵ &\quad v₀→Nv₁ ⏎
		&\quad v₁→Nv₀
	\end{array}
\end{equation*}
(Definitions not reachable from the initial variable are omitted.)
The resulting language is indeed regular:~$L(G)=❴N^{2k}C~|~k≥0 ❵$.
The grammar productions simulate subtyping against $\Delta$ during
derivation:
\begin{itemize}
	\item Variable $v_0=C \subtype x$ derives $C$ because if $x=C$, $C \subtype C$ holds.
	\item Variable $v_0=C \subtype x$ derives $N v_1$ because if $x=N x'$ for some type $x'$,
		the query in $v_0$ holds if $x'$ satisfies the query in $v_1$: \[
			C \subtype N x' ~\Leftarrow~ NNC \subtype Nx' ~\Leftarrow~ NC \suptype x' = v_1.
		\]
	\item Similarly, variable $v_1=NC \suptype x$ derives $N v_0$ because if $x=N x'$ for some type $x'$,
		the query in $v_1$ holds if $x'$ satisfies the query in $v_0$: \[
			NC \suptype Nx' ~\Leftarrow~ C \subtype x' = v_0.
		\]
	\item Nevertheless, $v_1=NC \suptype x$ cannot derive $C$, because the query $NC \suptype C$
		leads to an infinite proof:
		\begin{equation}\label{eq:non:terminating}
			NC \suptype C~⇐~NC \suptype NNC~⇐~C \subtype NC~⇐~…
		\end{equation}
\end{itemize}
As the infinite proof in the latter case is cyclic, it can be identified during
the grammar construction.

\paragraph*{Additional notations}
A query of the form~$t \sim x$, where~$t$ is a type,~$x$ is a type parameter, and~$\sim{}\!∈❴\subtype,=, \suptype❵$,
is denoted by~$q$.
A set of queries is denoted by~$Q$.
The set of~$Q$ queries where parameter~$x$ is substituted with term~$τ$ is denoted by~$Q[τ]$: \[
  Q[τ]=❴t \sim τ~|~t \sim x∈Q ❵
\]
The conjunction of the queries in~$Q$ is denoted by~$\bigwedge\!Q$.
We write~$\bigwedge\!Q \vdash \bigwedge\!Q'$ to say that the queries in~$Q$ imply the
queries in~$Q'$, i.e.,~$\bigwedge\!Q\Rightarrow\bigwedge\!Q'$.
If the queries in~$Q$ are true, we write~$\vdash\bigwedge\!Q$.
A set of queries is minimal if it does not contain redundant queries, e.g.,
including~$t \subtype x$ in addition to~$t = x$ although the latter implies
the former.

\begin{proof}[Proof of \cref{theorem:regf:upper}]
  Let~$Δ$ be a class table of type system \Tcm,
  and consider subtyping-queries between type~$t_⊥$ and types over
  super-alphabet~$Σ_⊤$.
  We show that the set of~$t_⊥$ supertypes is a regular forest,
  described by a regular tree grammar~$G=⟨Σ, V, v₀, R⟩$.

  \proofpart{Construction}
  Let~$Σ=Σ_⊤$.
  The set of variables~$V$ contains all sets of queries over~$\cl(t_⊥)$: \[
    V=❰Q~\left|~∀(t \sim{} x)∈Q,~t∈\cl(t_⊥),~
    \sim{}\!∈❴\subtype,=, \suptype❵\right.❱
\] Let~$v₀=❴t_⊥\subtype x ❵$.
  For every variable~$v=Q∈V$ and terminal~$σ∈Σ$,~$\rank(σ)=n$,
  add derivation rule~$v \produceσ(\vv v)$ to~$R$ if and only if there exists a sequence
  of variables~$\vv v=~\vv{Q}$, such that the sets in~$\vv Q$ are minimal and~$\bigwedge\!{\vvp{Q}}⊢\bigwedge\!{Q'}$,
  for set~$Q'$ and sets~$\vvp{Q}$,
  \begin{enumerate}
    \item~$Q'=Q[σ(\vv x)]=❴t \simσ(x_⊥,…,xₙ)~|~(t \sim x)∈Q ❵$
    \item~$∀ i=1,…,n,~Qᵢ'=Qᵢ[xᵢ]=❴t \sim xᵢ~|~(t \sim x)∈Qᵢ ❵$
  \end{enumerate}
  Expressed in words,
  parameter~$x$ in~$v$ is replaced with the term~$σ(x₁,…,xₙ)$,~$n=\rank(σ)$;
  the type parameters~$x$ in~$\vv v$ are enumerated according to their position in the sequence;
  then we require that the queries of~$\vv v$ prove all the queries in~$v$.
  As the inheritance and decomposition closure of~$t_⊥$ is always finite,
  variables~$\vv v$ can be found by traversing every possible proof of~$t \simσ(\vv x)$.

  \proofpart{Correctness of the construction}
  Observe that the grammar is well-defined.
  As~$Δ$ is non-expansive,~$\cl(t_⊥)$ is finite, thus~$V$ is finite.
  For every~$R$ production~$v \produceσ(\vv v)$, types~$t$ on the left-hand side of
  $v$ queries are drawn from~$\cl(t_⊥)$, so their sub-queries contain types also drawn from
  $\cl(t_⊥)$ by definition of the closure. Therefore,~$\vv v⊆V$, and thus~$R$ is finite.

  We show that variable~$v=Q$ derives tree~$t$,~$v→_G^* t$ if and only if substituting
  parameter~$x$ in~$Q$'s queries with type~$t$ satisfies them:
  \begin{equation}\label{eq:statement:2}
    v→_G^* t~⇔~⊢\bigwedge\!{v[t]}
  \end{equation}
  \cref{eq:statement:2} is proved by induction on the height of~$t$,~$n=\height(t)$.

  If~$n=0$, then~$t$ is a leaf,~$t=σ$. Variable~$v$, describing a set of
  queries~$Q$, produces~$σ$ if and only if~$v \produceσ∈R$.
  For nodes of rank~$0$, the construction introduces transition~$v \produceσ$
  if and only if the queries in set~$Q'=❴t \simσ~|~(t \sim x)∈Q ❵$
  are satisfied (by empty~$\vvp{Q}$), i.e., if replacing parameter~$x$ in~$Q$ with
  $t=σ$ satisfies its queries: \[
    v→_G^*σ~⇔~v \produceσ∈R~⇔~
⊢\bigwedge\!{v[σ]}
\] Suppose \cref{eq:statement:2} holds for~$n≥0$,
  and let~$t$ be a tree of height~$n+1$,~$t=σ(\vv t)$,~$σ∈Σ$,~$\vv t∈Σ^▵$,
  where the heights of~$\vv t$ are, at most,~$n$.
  Variable~$v$ derives~$t$,~$v→_G^* t$ if and only if~$v$ derives~$σ$ with production~$v \produceσ(\vv v)∈R$,
  where variables~$\vv v$ derive~$σ$'s children,~$\vv v→_G^* \vv t$.
  According to the construction,~$v \produceσ(\vv v)∈R$ if and only if~$\bigwedge\!{\vv{v}[\vv x]}⊢\bigwedge\!{v[σ(\vv x)]}$.
  By the inductive assumption,~$⊢\bigwedge\!{\vv v[\vv t]}$ if and only if~$\vv v→_G^* \vv t$,
  as the heights of types in~$\vv t$ are not higher than~$n$.
  Thus, by assigning~$\vv t$ to~$\vv x$ we get that~$⊢\bigwedge\!{\vv{v}[\vv t]}⊢\bigwedge\!{v[σ(\vv t)]}$,
  or~$⊢\bigwedge\!{v[σ(\vv t)]}$ in short: \[
    v→_G^*σ(\vv t)~⇔~v \produceσ(\vv v)∈R~∧~
    \vv v→_G^* \vv t~⇔~
⊢\bigwedge\!{v[σ(\vv t)]}
\] By applying \cref{eq:statement:2} to~$v₀=❴t_⊥\subtype x ❵$, we get that \[
    v₀→_G^* t~⇔~⊢\bigwedge\!{❴t_⊥\subtype t ❵}~⇔~t_⊥\subtype t,
\] and, therefore,~$L(G)=L_{t_⊥}(Δ)$ by definition.
\end{proof}

From \cref{theorem:regf:lower,theorem:regf:upper}, it follows that type systems \Tm and \Tcm
are precisely as complex as regular forests.
\begin{corollary}\label{theorem:regf}
  Class tables of non-expansive type systems describe regular forests: \[
𝓛(\Tm)=𝓛(\Tcm)=\|NRF|
\] \end{corollary}
\begin{proof}
  Consider that
  \1~$\|NRF|⊆𝓛(\Tm)$ by \cref{theorem:regf:lower},
  \2~$L(\Tm)⊆𝓛(\Tcm)$ as any \Tm class table conforms to \Tcm,~$\Tm⊂\Tcm$, and that
  \3~$𝓛(\Tcm)⊆\|NRF|$ by \cref{theorem:regf:upper};
  thus, \[
    \|NRF|⊆𝓛(\Tm)⊆𝓛(\Tcm)⊆\|NRF|~⇒~
𝓛(\Tm)=𝓛(\Tcm)=\|NRF| \qedhere
\] \end{proof}

When both expansive inheritance and multiple instantiation inheritance are taken away,
the complexity of the type system is trimmed down to deterministic forests.
\begin{proposition}\label{theorem:dregf}
  Class tables of non-expansive type systems with single instantiation inheritance
  describe deterministic regular forests: \[
𝓛(\Tbot)=𝓛(\Tc)=\|DRF|
\] \end{proposition}
\begin{proof}
  By repeating the constructions in the proofs of \cref{theorem:regf:lower} and \cref{theorem:regf:upper}
  without multiple instantiation inheritance, we get the corresponding results with deterministic
  regular forests (and grammars).

  In the proof of \cref{theorem:regf:lower}, a production~$v \produceσ(\vv v)$ is directly
  encoded as the inheritance rule~$v:σ(\vv v)$.
  Thus, the determinism of the grammar coincides with the class table's single instantiation inheritance: \[
    \infer{
      v:σ(\vv{v₁}),σ(\vv{v₂})~⇒~\vv{v₁}=\vv{v₂}
    }{
      \begin{array}{l}
        v \produceσ(\vv{v₁})~∧~v \produceσ(\vv{v₂})~⇔~v:σ(\vv{v₁}),σ(\vv{v₂}) ⏎
        v \produceσ(\vv{v₁})~∧~v \produceσ(\vv{v₂})~⇒~\vv{v₁}=\vv{v₂}
      \end{array}
    }
\] If a class table employs single instantiation inheritance, then the subtyping algorithm,
  consisting of the \textsc{Var} and \textsc{Super} rules, becomes deterministic.
  In the proof of \cref{theorem:regf:upper}, a production~$v \produceσ(\vv v)$ is
  introduced if and only if the queries encoded by~$\vv v$ prove the queries of~$v$
  when node~$σ$ is encountered.
  As subtyping is deterministic, at most one minimal such sequence of queries~$\vv v$ exists.
\end{proof}

\subsection{Non-Contravariant Subtyping Is Context-Free}
\label{section:non:contravariant}

Yet another way to make subtyping decidable is by prohibiting contravariant type
parameters.
Computationally, this restriction is less severe than non-expansive inheritance.
We show that non-contravariant class tables are as expressive as context-free forests
with grammars in GNF.
Recall that Greibach productions are of the form~$v(\vv x) \produce σ(\vv{τ})$,
i.e., where the root of the derived tree form is terminal,~$σ∈Σ$.

In this section, we use the extended variant of context-free grammars, defined in \cref{definition:ecftg},
where the initial tree form can be any tree---not just a variable.
Recall that in \cref{lemma:ecftg} we proved that extended context-free tree grammars are
equivalent in expressiveness to standard context-free grammars.

\begin{theorem}\label{theorem:cff:gnf:lower}
  Any context-free tree grammar in GNF
  can be encoded in a non-contravariant class table: \[
    \|NCF\sub{GNF}|⊆𝓛(\Txm)
\] \end{theorem}

\paragraph*{Intuition}
As in the regular case (\cref{theorem:regf:lower}), we want to use subtyping
to simulate grammar derivation.
Also, as before, we directly encode productions as inheritance rules and use
covariance to enable recursive subtyping.
In this case, however, we harness expansively-recursive inheritance to encode
the slightly more complex structure of context-free productions.

\paragraph*{Example}
Consider our palindrome running example.
The context-free grammar productions in \cref{eq:example:cff:grammar} (shown in \cref{section:forests}),
together with the initial tree $t_0=v_0E$, derive palindromes over $\{a,b\}$ terminated
by end-marker $E$.
We already encoded this grammar as a program in abstract syntax in
\cref{eq:example:program} (shown in \cref{section:subtyping})
and in \CSharp in \cref{lst:example:cff:program} (shown in \cref{section:aa}).
Notice that the parameterized variable $v_0(x)$ is encoded by the polymorphic
\inline[style=csharp]{interface v0<x>}.

\begin{proof}[Proof of \cref{theorem:cff:gnf:lower}]
  Let~$G=⟨Σ, V, t₀, R⟩$ be an extended context-free tree grammar in GNF.
  We encode grammar~$G$ as class table~$Δ$,~$\Txm⊢Δ$, such that the language of
  $Δ$ (over~$t_⊥=t₀$) is exactly the extent of the grammar,~$L(G)$.

  \proofpart{Construction}
  Class names are drawn from~$Γ=Σ∪V$. The rank of each type~$γ∈Σ∪V$
  is the same as its rank in the grammar. The parameters of class~$γ$ are covariant if~$γ∈Σ$,
  or invariant if~$γ∈V$.
  The subtype in queries is~$t_⊥=t₀$, and supertype names are drawn from~$Σ_⊤=Σ$.
  For every production~$v(\vv x) \produce σ(\vv{τ})∈R$, introduce inheritance rule
  $v(\vv{∘ x}):σ(\vv{τ})$ to~$Δ$.

  \proofpart{Correctness of the construction}
  We show that tree form~$t$ derives terminal tree~$t'∈Σ^▵$ if and only if type~$t$ is a subtype of type~$t'$:
  \begin{equation}\label{eq:statement:3}
    t→_G^* t'~⇔~t \subtype t',~t'∈Σ^\triangle
  \end{equation}
  Proof is by induction on the height of~$t'$,~$n=\height(t')$.

  If~$n=0$, then~$t'$ is a leaf,~$t'=σ$.
  Tree form~$t$ derives~$σ$ only if~$t=σ$ (in which case,~$t \subtypeσ$ is immediate),
  or if the root of~$t$ is a variable~$v$, and there is a production~$v(\vv x) \produceσ∈R$.
  In the latter case, class~$v$ inherits class~$σ$, so type~$t$ is still a subtype of~$σ$: \[
    v(\vv t)→_G^*σ~⇔~v(\vv x) \produceσ∈R~⇔~
    v(\vv{∘ x}):σ∈Δ~⇔~v(\vv t) \subtype \sigma
\] Assume \cref{eq:statement:3} holds for all trees of height~$n$,
  and let us prove this for tree~$t'=σ(\vvp{t})$ of height~$n+1$.
  Let~$t=ξ(\vv t)$.
  If the root of~$t$ is a terminal node~$ξ∈Σ$, then it must to be~$σ$ to derive~$t'$.
  To complete the derivation, we require the children of~$t$ to derive the children of~$t'$,
  $\vv t→_G^* \vvp{t}$. In this case, type~$σ(\vv t)$ is a subtype of type
  $σ(\vvp{t})$, by applying the \textsc{Var} rule and~$σ$'s covariance to
  get~$\vv t \subtype \vvp{t}$, which is true by the inductive assumption, given that~$\height(\vvp{t})\le n$: \[
σ(\vv t)→_G^*σ(\vvp{t})~⇔~\vv t→_G^* \vvp{t}~⇔~
    \vv t \subtype \vvp{t}~⇔~σ(\vv t) \subtypeσ(\vvp{t})
\] Otherwise the root of~$t$ is a variable,~$ξ=v∈V$.
  To derive~$t'$, there must be a production~$ρ=v(\vv x) \produceσ(\vv{τ})∈R$,
  and its application to~$t$ has to produce~$\vvp{t}$, \[
    \vv{τ}[\vv x←\vv t]→_G^* \vvp{t}.
\] Given production~$ρ$, the construction introduces the corresponding
  inheritance rule~$v(\vv{∘ x}):σ(\vv{τ})∈Δ$, so type~$t$ is a subtype of type~$t'$:
  By applying the \textsc{Super} rule, replacing~$v(\vv t)$ with~$σ(\vv{τ}[\vv x←\vv t])$,
  and then the \textsc{Var} rule, to reduce the query to~$\vv{τ}[\vv x←\vv t] \subtype \vvp{t}$,
  which is true by the inductive assumption since~$\height(\vvp{t})\le n$: \[
    \begin{aligned}
      v(\vv t)→_G^*σ(\vvp{t})~&⇔~
      v(\vv x) \produceσ(\vv{τ})~∧~\vv{τ}[\vv x←\vv t]→_G^* \vvp{t} ⏎
      &⇔~v(\vv x):σ(\vv{τ})~∧~τ[\vv x←\vv t] \subtype \vvp{t} ⏎
      &⇔~v(\vv t) \subtypeσ(\vvp{t})
    \end{aligned}
\] By applying \cref{eq:statement:3} to~$t_⊥=t₀$, we get that~$t₀→_G^* t~⇔~t_⊥\subtype t$;
  thus,~$L(G)=L_{t_⊥}(Δ)$ by definition.
\end{proof}

\begin{theorem}\label{theorem:cff:gnf:upper}
  Non-contravariant class tables describe context-free grammars in GNF: \[
𝓛(\Txm)⊆\|NCF\sub{GNF}|
\] \end{theorem}

\paragraph*{Intuition}
For starters, assume that all the type parameters in the class table are covariant
(no invariance).
Consider the subtyping query $t \subtype t'$ where $t=\gamma(\vv{t})$ and $t'=\gamma'(\vv{t'})$.
To prove this query, we first use the \textsc{Super} rule to replace $t$ with $\gamma'(\vv{\tau}[\vv{x}\leftarrow\vv{t}])$,
assuming that class $\gamma(\vv{x})$ transitively inherits ($:^*$) the type $\gamma'(\vv{\tau})$.
In the next step, we use the \textsc{Var} rule to recursively validate the children of $\gamma'$,
$\vv{\tau}[\vv{x}\leftarrow\vv{t}] \subtype \vv{t'}$.
To simulate subtyping with grammar derivation, we simply encode transitive
inheritance as productions, i.e., if $\gamma(\vv{x}) :^* \gamma'(\vv{\tau})$,
$\gamma(\vv{x}) \produce \gamma'(\vv{\tau})$.
Then, the query $t \subtype t'$ is true if and only if $t \rightarrow^* t'$
because, as in the subtyping case, $t=\gamma(\vv{t})$ derives $\gamma'(\vv{\tau}[\vv{x}\leftarrow\vv{t}])$
and then $\vv{\tau}[\vv{x}\leftarrow\vv{t}]$ recursively derive $\vv{t'}$: \[
	\begin{tabular}{l@{}l@{}l}
		$\gamma(\vv{t}) \subtype \gamma'(\vv{t'})$ &$~\Leftrightarrow~ \gamma(\vv{x}) :^* \gamma'(\vv{\tau})$ &
		$~\wedge~ \vv{\tau}[\vv{x}\leftarrow\vv{t}] \subtype \vv{t'}$ \\
		$\gamma(\vv{t}) \rightarrow^* \gamma'(\vv{t'})$ &$~\Leftrightarrow~ \gamma(\vv{x}) \produce \gamma'(\vv{\tau})$ &
		$~\wedge~ \vv{\tau}[\vv{x}\leftarrow\vv{t}] \rightarrow^* \vv{t'}$
	\end{tabular}
\]

Nonetheless, this clean correspondence between subtyping and context-free grammars
collapses when invariance is re-introduced to the class table.
For instance, consider the following inheritance rules:
\begin{equation}\label{eq:example:class:table:2}
  \begin{aligned}
    a(+x,∘ y) &:∅⏎
    b(∘ z) &: a(z, z)
  \end{aligned}
\end{equation}
In the query~$b(t) \subtype a(t₁, t₂)$, we use inheritance and decomposition to
deduce both~$t \subtype t₁$ and~$t=t₂$.
If type~$t$ were to be encoded directly as a tree form, as before, then it should be able to derive
the two trees~$t₁$ and~$t₂$, while simulating both covariance and invariance simultaneously.

To reproduce this behavior in a context-free tree grammar, we encode the kind of position
of type $t$ in the tree form that represents it.
Let $\gamma$ be the root of $t$.
If $t$ appears in a covariant position, then $\gamma$ is encoded by variable $\gamma_+$,
or otherwise by variable $\gamma_\circ$ in an invariant position.
If $t$ is in a covariant position, $t \subtype t_1$, variable $\gamma_+$ derives
the tree $t'$ if and only if $t' \subtype t_1$;
otherwise, if $t$ is in an invariant position, $t = t_2$, variable $\gamma_\circ$ derives
the tree $t'$ if and only if $t' = t_2$.

Variable $\gamma_\circ$ simply derives the node $\gamma$.
If type $t$ is in an invariant position, $t=t_2$, we encode $t$
by annotating each of its nodes with $\circ$, yielding the tree form $t_\circ$.
This tree form derives, recursively, the tree $t$ exactly, so \[
	t_\circ \rightarrow^* t' ~\Leftrightarrow~ t'=t \quad\Rightarrow\quad t_\circ \rightarrow^* t_2 ~\Leftrightarrow~ t=t_2
\] as required.
On the other hand, if $t$ appears in a covariant position, $t \subtype t_1$,
then its children may appear in either subtyping or equivalence queries (or even both simultaneously).
Let $t=\gamma(\vv{t})$ and $t_i$ be some child of $\gamma$.
We encode $t_i$ as two tree forms, ${t_i}_+$ and ${t_i}_\circ$.
Then, we can use ${t_i}_+$ whenever $t_i$ is covariant, or otherwise ${t_i}_\circ$ when $t_i$
is invariant.
Thus, type $t=\gamma(\vv{t})$ is encoded as the tree form $t_+=\gamma_+(\vv{t_+},\vv{t_\circ})$,
where children $\vv{t_+}$ are, recursively, encoded in a similar way, and $\vv{t_\circ}$ are
encoded as mentioned above.

Continuing our example, we get that class~$b$ in \cref{eq:example:class:table:2} is
encoded by the following productions: \[
  b₊(z₊, z_∘) \produce a(z₊, z_∘) \qquad b_∘(z_∘) \produce a(z_∘, z_∘)
\] Variable~$b₊$ encodes class~$b$ in covariant positions, so the first argument to type~$a$
is in a covariant position, and~$z₊$ is used, while the second position is invariant, and~$z_∘$ is
picked instead.
Parameter~$z_∘$ in variable~$b_∘$ fits both positions, as the initial position (of root~$a$)
is invariant.

\paragraph*{Additional notations}
To encode unground type~$τ$ as invariant tree form~$τ_∘$, add
annotation~$∘$ to all its nodes:
\begin{equation}\label{eq:invariant:encoding}
τ_∘=\begin{cases}
    x_∘ &τ=x ⏎
γ_∘(\vv{τ_∘}) &τ=γ(\vv{τ})
  \end{cases}
\end{equation}
We get the covariant tree form~$τ₊$ by encoding both the covariant and invariant
versions of its sub-trees:
\begin{equation}\label{eq:covariant:encoding}
τ₊=\begin{cases}
    x₊ &τ=x ⏎
γ₊(\vv{τ₊}, \vv{τ_∘}) &τ=γ(\vv{τ})
  \end{cases}
\end{equation}
To decode an invariant tree form~$τ_∘$ into the tree it represents~$τ$,
simply remove all the~$∘$ annotations from its nodes.
If the tree form is covariant~$τ₊$, also remove any~$+$ annotations,
and trim-out half of the sub-trees:
\begin{equation}\label{eq:invariant:covariant:decoding}
τ=\begin{cases}
    x &τ_∘=x_∘~∨~τ₊=x₊ ⏎
γ(\vv{τ}) &τ_∘=γ_∘(\vv{τ_∘})~∨~τ₊=γ₊(\vv{τ₊}, \vv{τ_∘})
  \end{cases}
\end{equation}

Observe that if variables~$\vv{x₊}$ in covariant tree form~$τ₊$ are substituted with covariant trees~$\vv{t₊}$,
and the invariant variables~$\vv{x_∘}$ are substituted with the corresponding invariant trees~$\vv{t_∘}$,
then the resulting tree~$τ₊[ \vv{x₊}←\vv{t₊},\vv{x_∘}←\vv{t_∘} ]$
covariantly encodes type~$τ[ \vv x←\vv t ]$. The same holds for invariant tree forms:
\begin{equation}\label{eq:auxiliary:lemma}
τ_◇[ \vv{x₊}←\vv{t₊},\vv{x_∘}←\vv{t_∘} ]=(τ[ \vv x←\vv t ])_◇,
~◇∈❴+,∘ ❵
\end{equation}
(The proof is technical, by induction on the height of~$τ_◇$.)

\begin{proof}[Proof of \cref{theorem:cff:gnf:upper}]
  Given class table~$Δ$,~$\Txm⊢Δ$, fixed subtype~$t_⊥$, and
  super-alphabet~$Σ_⊤$, we construct extended context-free tree
  grammar~$G=⟨Σ, V, t₀, R⟩$ that describes the same language
  as~$Δ$,~$L(G)=L_{t_⊥}(Δ)$.

  In this proof we use covariant and invariant tree encoding of types,
  defined in \cref{eq:invariant:encoding,eq:covariant:encoding,eq:invariant:covariant:decoding}.

  \proofpart{Construction}
  Let~$Σ=Σ_⊤$.
  For every class name~$γ∈Γ$, add variables~$γ₊$ and~$γ_∘$ to~$V$, where~$\rank(γ_∘)=\rank(γ)$
  but~$\rank(γ₊)=2·\rank(γ)$.
  Let~$t₀=(t_⊥)₊$.
  For each terminal type~$σ∈Σ_⊤$, add rule~$σ_∘(\vv{x_∘}) \produceσ(\vv{x_∘})$
  to~$R$.
  For each class name~$γ∈Γ$ and terminal type~$σ∈Σ_⊤$ of rank~$n$, if
  class~$γ$ inherits (or, is)~$σ$,~$γ(\vv x):^*σ(\vv{τ})$,
  introduce the following production to~$R$: \[
γ₊(\vv{x₊}, \vv{x_∘}) \produceσ(\vv{τ_{\vof[\relax]{σ}}})=σ({(τ₁)}_{\vof[1]{σ}},…,{(τₙ)}_{\vof[n]{σ}})
\] (Recall that~$\vof{σ}$ denotes the sequence of parameter variances of class~$σ$,
  and~$\vof[i]{σ}$ denotes the variance of the \nth{i} parameter.)
  As inheritance is acyclic, transitive inheritance~$:^*$ is finite, so the construction of~$R$ must terminate.

  \proofpart{Correctness of the construction}
  Each variable~$σ_∘$, for terminal node~$σ∈Σ$, directly derives~$σ$.
  Therefore, any invariant tree form~$t_∘$ over~$Σ$,~$t∈Σ^▵$ derives the
  terminal tree~$t$, and only it:
  \begin{equation}\label{eq:fact:1}
    t_∘→_G^* t',~t∈Σ^▵~⇔~t=t'
  \end{equation}
  Conversely, let~$t₊=γ₊(\vv{t₊},\vv{t_∘})$ be a covariant tree form.
  We show that~$t₊$ derives tree~$t'$ if and only if type~$t$ is a subtype of type~$t'$:
  \begin{equation}\label{eq:statement:4}
    t₊→_G^* t'~⇔~t \subtype t'
  \end{equation}
  \cref{eq:statement:4} is proved by induction on the height of~$t'$,~$n=\height(t')$.

  If~$n=0$, then~$t'$ is a leaf,~$t'=σ$.
  The root of~$t₊$, variable~$γ₊$, derives~$σ$ if and only if~$γ=σ$,
  or if class~$γ$ inherits class~$σ$,~$γ(\vv x):^*σ$; in both cases
  the construction introduces production~$γ₊(\vv{x₊},\vv{x_∘}) \produceσ$ to~$R$: \[
γ₊(\vv{t₊},\vv{t_∘})→_G^*σ~⇔~γ₊(\vv{x₊},\vv{x_∘}) \produceσ∈R~⇔~
γ(\vv{◇ x}):^*σ~⇔~γ(\vv t) \subtype \sigma
\] Assume \cref{eq:statement:4} holds for~$n≥0$, and let us prove it for tree
  $t'$ of height~$n+1$,~$t'=σ(\vvp{t})$.
  As grammar~$G$ is in GNF, tree form~$t₊$ can
  derive terminal tree~$t'$ if and only if
  \1~there is an~$R$ production that derives~$γ₊$, the root of~$t₊$, to~$σ$, the root of~$t'$,~$
γ₊(\vv{x₊},\vv{x_∘}) \produceσ(\vv{τ_{\vof[\relax]{σ}}})∈R,
  $ such that \2~the resulting sub-tree forms
  $\vv{τ_◇}$, substituted with~$\vv{t₊}$ and~$\vv{t_∘}$, derive~$σ$'s children,~$\vvp{t}$:
  \begin{equation}\label{eq:part:1}
    \begin{aligned}
      &γ₊(\vv{t₊},\vv{t_∘})→_G^*σ(\vvp{t})~⇔⏎
      &\qquadγ₊(\vv{x₊},\vv{x_∘}) \produceσ(\vv{τ_{\vof[\relax]{σ}}})∈R~∧~
      \vv{τ_{\vof[\relax]{σ}}}[\vv{x₊}←\vv{t₊},\vv{x_∘}←\vv{t_∘}]→_G^* \vvp{t}
    \end{aligned}
  \end{equation}

  Now focus on tree forms~$\vv{τ_{\vof[\relax]{σ}}}[\vv{x₊}←\vv{t₊},\vv{x_∘}←\vv{t_∘}]$
  in \cref{eq:part:1}.
  Following the observation of \cref{eq:auxiliary:lemma}, these forms encode the unground types~$\vv{τ}$
  assigned with types~$\vv t$, while respecting the original variances of~$σ$,~$\vof{σ}$: \[
    \vv{τ_{\vof[\relax]{σ}}}[\vv{x₊}←\vv{t₊},\vv{x_∘}←\vv{t_∘}]=
    (\vv{τ}[\vv x←\vv t])_{\vof{σ}}
\] Let us then simplify \cref{eq:part:1} by writing:
  \begin{equation}\label{eq:part:2}
    (\ref{eq:part:1})~⇔~
γ₊(\vv{x₊},\vv{x_∘}) \produceσ(\vv{τ_{\vof[\relax]{σ}}})∈R~∧~
    (\vv{τ}[\vv x←\vv t])_{\vof{σ}}→_G^* \vvp{t}
  \end{equation}

  Focus on tree forms~$(\vv{τ}[\vv x←\vv t])_{\vof{σ}}$ in \cref{eq:part:2}.
  If tree form~$(τᵢ[\vv x←\vv t])_{\vof[i]{σ}}$ is invariant,
  $\vof[i]{σ}=∘$, then by \cref{eq:fact:1},
  \begin{equation}\label{eq:part:3}
    (τᵢ[\vv x←\vv t])_∘→_G^* {t'}ᵢ~⇔~
τᵢ[\vv x←\vv t]={t'}ᵢ
  \end{equation}
  Otherwise the tree form is covariant,~$\vof[i]{σ}=+$.
  As~${t'}ᵢ$ is a sub-tree of~$t'$, its height is~$n$ at most, so we can apply the inductive assumption
  to get
  \begin{equation}\label{eq:part:4}
    (τᵢ[\vv x←\vv t])₊→_G^* {t'}ᵢ~⇔~τᵢ[\vv x←\vv t] \subtype {t'}ᵢ
  \end{equation}

  Production~$γ₊(\vv{x₊},\vv{x_∘}) \produceσ(\vv{τ_{\vof[\relax]{σ}}})$ is induced by
  the construction following the inheritance rule~$γ(\vv{◇ x}):^*σ(\vv{τ})$.
  Thus, type~$γ(\vv t)$ is a subtype of~$σ(\vv{τ}[\vv x←\vv t])$ by inheritance,
  which is a subtype of~$t'=σ(\vvp{t})$ by decomposition.
  In \cref{eq:part:3,eq:part:4}, we showed that child types~$\vv{τ}[\vv x←\vv t]$ and~$\vvp{t}$
  respect~$σ$'s variances, so decomposition can be applied.
  Let us extend \cref{eq:part:2}:
  \begin{equation}\label{eq:part:5}
    \begin{aligned}
      (\ref{eq:part:2})~&⇔~γ(\vv{◇ x}):^*σ(\vv{τ})~∧~∀ i,~\begin{cases}
τᵢ[ \vv x←\vv t ]={t'}ᵢ & \vof[i]{σ}=∘ ⏎
τᵢ[ \vv x←\vv t ] \subtype {t'}ᵢ & \vof[i]{σ}=+
      \end{cases} ⏎
      &⇔~γ(\vv t) \subtypeσ(\vvp{t})
    \end{aligned}
  \end{equation}
  Overall, from \cref{eq:part:1,eq:part:2,eq:part:5}, we have \[
γ₊(\vv{t₊},\vv{t_∘})→_G^*σ(\vvp{t})~⇔~γ(\vv t) \subtypeσ(\vvp{t}),
\] concluding the induction.

  Finally, we can apply \cref{eq:statement:4} to~$t₀={(t_⊥)}₊$ to get
  $t₀→_G^* t~⇔~t_⊥\subtype t$; thus,~$L(G)=L_{t_⊥}(Δ)$
  by definition.
\end{proof}

From \cref{theorem:cff:gnf:lower,theorem:cff:gnf:upper}, it follows that the computational class
of type system \Txm is described by context-free tree grammars in GNF.
\begin{corollary}\label{theorem:cff:gnf}
  Non-contravariant class tables describe context-free forests
  that have grammars in GNF: \[
𝓛(\Txm)=\|NCF\sub{GNF}|
\] \end{corollary}
\begin{proof}
  $\|NCF\sub{GNF}|⊆𝓛(\Txm)$ by \cref{theorem:cff:gnf:lower};
  $𝓛(\Txm)⊆\|NCF\sub{GNF}|$ by \cref{theorem:cff:gnf:upper};
  $\|NCF\sub{GNF}|⊆𝓛(\Txm)⊆\|NCF\sub{GNF}|~⇒~𝓛(\Txm)=\|CFF\sub{GNF}|$.
\end{proof}

Type system \Tx is the subset of \Txm employing single-instantiation inheritance.
As before (\cref{theorem:dregf}), single instantiation inheritance reduces the complexity of
the type system to its deterministic variant.
\begin{proposition}\label{theorem:dcff}
  Class tables of non-contravariant type systems with single instantiation inheritance
  describe deterministic context-free forests: \[
𝓛(\Tx)=\|DCF|
\] \end{proposition}
\begin{proof}
  Every deterministic context-free forest has a deterministic tree
  grammar in GNF~\cite{Guessarian:83}.
  Again, by repeating the proof of \cref{theorem:cff:gnf:lower,theorem:cff:gnf:upper}
  with single instantiation inheritance, we get the corresponding results regarding
  DCF.

  In the proof of \cref{theorem:cff:gnf:lower}, we encode each production~$v(\vv x) \produceτ$
  directly as the inheritance rule~$v(\vv{∘ x}) \produceτ$.
  Therefore, the grammar's determinism ensures the resulting class table does not use
  multiple instantiation inheritance: \[
    \infer{
      v(\vv{∘ x}):σ(\vv{τ₁}),σ(\vv{τ₂})~⇒~\vv{τ₁}=\vv{τ₂}
    }{
      \begin{array}{l}
        v(\vv x) \produceσ(\vv{τ₁})~∧~v(\vv x) \produceσ(\vv{τ₂})~⇔~
        v(\vv{∘ x}):σ(\vv{τ₁}),σ(\vv{τ₂}) ⏎
        v(\vv x) \produceσ(\vv{τ₁})~∧~v(\vv x) \produceσ(\vv{τ₂})~⇒~\vv{τ₁}=\vv{τ₂}
      \end{array}
    }
\] In the proof of \cref{theorem:cff:gnf:upper}, we derive \1~$γ_∘$ to~$σ$, which is
  always deterministic, and \2~$γ₊$ to~$σ(\vv{τ_{\vof[\relax]{σ}}})$ if and
  only if~$γ$ inherits~$σ(\vv{τ})$, which is also deterministic when the class table employs
  single instantiation inheritance: \[
    \infer{
γ₊(\vv{x₊},\vv{x_∘}) \produceσ(\vv{(τ₁)_{\vof[\relax]{σ}}})~∧~
γ₊(\vv{x₊},\vv{x_∘}) \produceσ(\vv{(τ₂)_{\vof[\relax]{σ}}})~⇔~
      \vv{(τ₁)_{\vof[\relax]{σ}}}=\vv{(τ₂)_{\vof[\relax]{σ}}}
    }{
      \begin{array}{l}
γ(\vv{◇ x}):^*σ(\vv{τ₁}),σ(\vv{τ₂})~⇔~
γ₊(\vv{x₊},\vv{x_∘}) \produceσ(\vv{(τ₁)_{\vof[\relax]{σ}}})~∧~
γ₊(\vv{x₊},\vv{x_∘}) \produceσ(\vv{(τ₂)_{\vof[\relax]{σ}}}) ⏎
γ(\vv{◇ x}):^*σ(\vv{τ₁}),σ(\vv{τ₂})~⇒~\vv{τ₁}=\vv{τ₂}
      \end{array}
    }\qedhere
\] \end{proof}

	\section{Treetop---a Context-Free API Generator}
	\label{section:generator}
	Upon its introduction~\cite{Fowler:2005}, fluent API was recognized as a promising method for
embedding domain-specific languages.
As fluent API methods are called in a chain, and not imperatively, they can
enforce a DSL syntax, or an API protocol in general, at compile time.
(See a brief introduction of the concept in \cref{section:fluent} and more detailed
discussions by \citet{Gil:19} and \citet{Yamazaki:2019}.)
The subject attracted academic interest, as
researchers introduced new fluent API techniques supporting increasingly
complex DSLs~\cite{Xu:2010,Gil:Levy:2016,Nakamaru:17,Gil:19,Yamazaki:2019,Nakamaru:2020,Gil:20}.
Because the study of fluent API is motivated by its practical applications,
results in the field are given in the form of actual tools: fluent API
generators.
A generator accepts a DSL grammar as input, and encodes it as a
fluent API in the target programming language.
The quality of a generator is determined by its domain, i.e.,
the class of languages (and grammars) for which it can generate fluent APIs.
Generated fluent APIs also need to be practical, in terms of code size and compilation time.
The last two significant works on fluent APIs, published by \citet{Gil:19} and \citet{Yamazaki:2019},
presented the generators Fling and TypeLevelLR
that support all deterministic context-free languages (DCFLs). (\citet{Gil:20} later proved that
the methods implemented in these generators cannot support languages beyond DCFL.)
Programs generated by Fling and TypeLevelLR were shown to compile in linear time,
making both tools practical solutions for embedding DSLs.

The idea of integrating subtyping into fluent APIs was first proposed by \citet{Grigore:2017}.
In his design, the role of the fluent API becomes merely to encode the input DSL program in a type,
so it can be accepted by a subtyping machine.
Grigore's subtyping machine implementation of the CYK parser
was introduced as a part of a fluent API generator that supports all context-free languages.
Moving up from DCFL to CFL is a welcome improvement:
CFLs are described by general context-free grammars, which are more expressive than
their deterministic variants (LR, LL, LALR, etc.) and are easier to compose, as they
produce no conflicts.
Unfortunately, it turned out that Grigore's generator is not practical.
The APIs it generates take a lot of time to compile, and cause the compiler to
crash with stack overflow even for basic examples.

\cref{theorem:cff:gnf:lower} presented in \cref{section:expressiveness} states that any tree grammar in
GNF can be recognized by a subtyping machine with expansive inheritance and
multiple instantiation inheritance.
As all context-free languages have grammars in GNF~\cite{Greibach:65}, this
means we just found a new method for encoding CFLs using subtyping that, in contrast to Grigore's
construction, does not employ contravariance.
Nevertheless, can this method be used in a real programming language, such as~\CSharp?
Unlike \Java,~\CSharp endorses multiple instantiation inheritance
(required by our construction)---but its version of this feature is weaker than the one
used in this paper.
In~\CSharp, if class~$γ$ has (transitive) supertypes~$τ$ and~$τ'$, then there must
not be a substitution~$s$ that unifies~$τ$ and~$τ'$:
\begin{equation}\label{eq:unification}
\begin{aligned}
  \infer{γ:⁺τ'~\checkmark}{τ \not= τ' \wedge γ:⁺τ ~\Rightarrow~ ∄ s,~τ[s]=τ'[s]}
\end{aligned}
\end{equation}
(This trait of~\CSharp was also discussed by \citet{Kennedy:Pierce:07}.)
It is possible that \cref{theorem:cff:gnf:lower} cannot be utilized to its full
potential in~\CSharp, as the construction described in its proof does not
conform to the restriction of \cref{eq:unification} in the general case.
Nevertheless, the construction does work for all CFLs. If we encode a context-free
production~$vx::=αx$ as the inheritance rule~$v(x):αx$, then unification with another rule~$v(x):α'x$ can
only occur for~$α'=α$. To prevent this, we simply ignore duplicated productions.

Thus, \cref{theorem:cff:gnf:lower} paved the way for the development of our new API generator,
Treetop, used to embed CFLs in~\CSharp†{The Treetop project is available at \url{https://github.com/OriRoth/treetop}.}.
Treetop is a~\CSharp \emph{source generator}, a project preprocessor that creates source
files, similarly to \Java's annotations processors.
Source generation is an experimental Visual Studio feature for~\CSharp
introduced in 2020†{For more details, see \url{https://devblogs.microsoft.com/dotnet/introducing-c-source-generators/}.}.
After an initial deployment, Treetop automatically transforms DSL grammars
into DSL APIs during the standard compilation cycle.
Given a context-free grammar~$G$ in file \texttt{MyDSL.cfg} describing a DSL, Treetop
converts it into an equivalent grammar in GNF,
constructs class table~$Δ∈\Txm$ that recognizes the DSL,~$L(Δ)=L(G)$, and
encodes the result as~\CSharp interfaces in a new source file \texttt{MyDSLAPI.cs}.
On request, Treetop also generates a subtyping-based fluent API for the DSL, following Grigore's design.

In theory, the APIs generated by Treetop should be efficient in both space (code size)
and time (compilation time of API invocations).
The class table construction described in the proof of \cref{theorem:cff:gnf:lower}
is straightforward, so the crucial step of the encoding is the
conversion of the input grammar to GNF.
This transformation is known to introduce a polynomial increase in the size of
the grammar~\cite{Blum:99}, which is translated to a similar increase in code size.
An API generated by Treetop is invoked in subtyping queries, whose compilation
time depends on the compiler implementation.
A specialized compiler can reduce Treetop's interfaces and queries
back to CFGs and strings, whose membership problem is solved in
polynomial time~\cite{Valiant:75b}.

\subsection{Treetop Use Case: The Canvas API}
\label{section:use}
\emph{Canvas} is an Android library for drawing two-dimensional
graphics\footnote{The Canvas API is described in \url{https://developer.android.com/reference/android/graphics/Canvas}.}.
Besides its various sketching methods, Canvas makes it possible to \emph{save} and \emph{restore} the current
image layout, with the restriction that ``it is an error to call \cc{Restore()} more times than \cc{Save()} was called''.
\citet{Ferles:20} found that this restriction institutes a context-free protocol
in the Canvas interface.
Let us demonstrate Treetop's capabilities by using it to create a Canvas API that enforces this protocol
at compile time.

First, we describe the Canvas API with a context-free grammar.
To make things simple, we use the \cc{Draw} token to represent any Canvas method other
than \cc{Save} and \cc{Restore}:
\begin{equation}\label{eq:canvas:grammar}
\begin{array}{ll}
	\text{(a)~}\-Canvas- \produce \+Draw+~\-Canvas- \qquad &\text{(b)~}\-Canvas- \produce \+Save+~\-Canvas-~\+Restore+~\-Canvas- \\
	\text{(c)~}\-Canvas- \produce \+Save+~\-Canvas- \qquad &\text{(d)~}\-Canvas- \produce \varepsilon
\end{array}
\end{equation}
The Canvas protocol is manifested in rule (b),
which assures that any call to \cc{Restore} is preceded by, and matched with a prior call to \cc{Save}.
In addition, rule (c) makes it possible to save the picture without restoring it later on.
We encode the formal Canvas grammar of \cref{eq:canvas:grammar} in the file
\texttt{Canvas.cfg}, as shown in \cref{lst:grammar}.
The line \inline{Canvas ::=} represents the $\varepsilon$ production (d).
\begin{code}[backgroundcolor=\color{perfect-grey},caption={Treetop grammar file \texttt{Canvas.cfg}, encoding the Canvas grammar in \cref{eq:canvas:grammar}},label={lst:grammar}]
Canvas ::= Draw Canvas
Canvas ::= Save Canvas Restore Canvas
Canvas ::= Save Canvas
Canvas ::=
\end{code}

When the project is re-compiled, Treetop detects the grammar in file \texttt{Canvas.cfg} 
and generates a new source file \texttt{CanvasAPI.cs} that
contains the Canvas subtyping machine and fluent API.
The content of \texttt{CanvasAPI.cs} is shown in \cref{lst:canvas:generated}.

\begin{code}[float=p,style=csharp,caption={Source file \texttt{CanvasAPI.cs} generated by Treetop, containing the Canvas subtyping
             machine and fluent API},label={lst:canvas:generated}]
namespace CanvasAPI {
  public interface Draw<out _x> {}
  public interface Restore<out _x> {}
  public interface Save<out _x> {}
  public interface Canvas2<_x> : Draw<_x>, Draw<Canvas3<_x>>, Restore<Save<_x>>,
    Restore<Save<Canvas3<_x>>>, Restore<Canvas<Save<_x>>>,
    Restore<Canvas<Save<Canvas3<_x>>>>, Save<_x>, Save<Canvas3<_x>> {}
  public interface Canvas<_x> : Draw<_x>, Draw<Canvas3<_x>>, Restore<Save<_x>>,
    Restore<Save<Canvas3<_x>>>, Restore<Canvas<Save<_x>>>,
    Restore<Canvas<Save<Canvas3<_x>>>>, Save<_x>, Save<Canvas3<_x>> {}
  public interface Canvas3<_x> : Draw<_x>, Draw<Canvas3<_x>>, Restore<Save<_x>>,
    Restore<Save<Canvas3<_x>>>, Restore<Canvas<Save<_x>>>,
    Restore<Canvas<Save<Canvas3<_x>>>>, Save<_x>, Save<Canvas3<_x>> {}
  public interface BOTTOM {}
  public interface Canvas : Canvas2<BOTTOM> {}
  namespace FluentAPI {
    public class Wrapper<T> {
      public readonly System.Collections.Generic.List<CanvasToken> values =
        new System.Collections.Generic.List<CanvasToken>();
      public Wrapper<T> AddRange<S>(Wrapper<S> other) {
        this.values.AddRange(other.values);
        return this;
      }
      public Wrapper<T> Add(CanvasToken value) {
        values.Add(value);
        return this;
      }
      public System.Collections.Generic.List<CanvasToken> Done<API>() where API : T {
        return values;
      }
    }
    public enum CanvasToken { Draw, Restore, Save, }
    public static class Start {
      public static Wrapper<Draw<BOTTOM>> Draw() {
        return new Wrapper<Draw<BOTTOM>>().Add(CanvasToken.Draw); }
      public static Wrapper<Draw<_x>> Draw<_x>(this Wrapper<_x> _wrapper) {
        return new Wrapper<Draw<_x>>().AddRange(_wrapper).Add(CanvasToken.Draw); }
      public static Wrapper<Restore<BOTTOM>> Restore() {
        return new Wrapper<Restore<BOTTOM>>().Add(CanvasToken.Restore); }
      public static Wrapper<Restore<_x>> Restore<_x>(this Wrapper<_x> _wrapper) {
        return new Wrapper<Restore<_x>>().AddRange(_wrapper).Add(CanvasToken.Restore); }
      public static Wrapper<Save<BOTTOM>> Save() {
        return new Wrapper<Save<BOTTOM>>().Add(CanvasToken.Save); }
      public static Wrapper<Save<_x>> Save<_x>(this Wrapper<_x> _wrapper) {
        return new Wrapper<Save<_x>>().AddRange(_wrapper).Add(CanvasToken.Save); }
      public static System.Collections.Generic.List<CanvasToken> Done<Canvas>() {
        return new System.Collections.Generic.List<CanvasToken>(); }
    }
  }
}
\end{code}

Treetop's encoding algorithm comprises three main steps.
First, Treetop \emph{reverses} the input grammar and converts it to GNF.
Recall that a CFG is reversed by reversing the right-hand side of each of its
productions---the reason why the grammar is reversed is explained below.
Second, Treetop applies the construction in the proof of \cref{theorem:cff:gnf:lower}
to convert the grammar into a subtyping machine (lines 2--15 of \cref{lst:canvas:generated}).
Third, Treetop constructs a fluent API (lines 16--49 of \cref{lst:canvas:generated}).
This fluent API utilizes the subtyping machine generated earlier to recognize the Canvas grammar.
A fluent chain of method calls starts with variable \cc{Start};
subsequent method calls encode the terminals of the input grammar, \+Draw+, \+Save+, and \+Restore+:
\[\inline{Start.Draw().Draw().Save().Draw().Restore().Save().Draw()$\ldots$}\]
The intermediate chain type is \cc{Wrapper<T>}, where type variable \cc{T}
captures the identities of the methods called thus far, e.g., \cc{T}=\cc{Draw<Save<Restore<$\ldots$>{}>{}>}.
When, for instance, method \cc{Draw()} is called, type \cc{Draw} is appended to \cc{T},
yielding the type \cc{Wrapper<Draw<T>{}>}.
Note that \cc{Wrapper}'s type argument records the method calls in reverse order,
i.e., the last method called is placed on the top of the type.
Treetop solves this issue by reversing the original input grammar $G$.
Then, a reversed type \cc{T} belongs to the reverse of $L(G)$
if and only if the fluent API method calls, in order of invocation, belong to $L(G)$.

A fluent chain is terminated by a call to method \cc{Done} of class \cc{Wrapper},
whose signature is
\[\inline[style=csharp]{List<CanvasToken> Done<API>() where API:T}\]
Recall that \cc{T} is substituted by a type describing the fluent chain.
When method \cc{Done} is called with type argument \cc{API}=\cc{Canvas}, declared in line 15
of \cref{lst:canvas:generated}, the method's constraint \cc{API:T} invokes the
subtyping query
\begin{equation}\label{eq:implied:query}
	\cc{Canvas} \subtype \cc{T}
\end{equation}
Type \cc{Canvas} is a part of the subtyping machine that encodes the initial grammar tree $t_0$.
Thus, the subtyping query in \cref{eq:implied:query} initiates the Canvas subtyping machine,
which accepts type \cc{T} if and only if the fluent chain that it encodes conforms to the Canvas
API protocol.
Finally, method \cc{Done} returns a list of tokens representing the fluent chain
that can be analyzed at runtime.

A use example of the Canvas fluent API is shown in \cref{lst:usage:example},
describing the content of another source file \texttt{UseExample.cs}.
The code demonstrates two invocations of the Canvas fluent API: the first
breaks the Canvas protocol, so it does not compile (line 5), and the second
adheres to the protocol, so it compiles properly (lines 6--7).
The loop in line 8 simply lists the tokens returned from the chain above,
printing to the console
\[\inline{Draw Draw Save Draw Restore Draw Save Draw Draw}\]

\begin{code}[float=h,style=csharp,caption={Source file \texttt{UseExample.cs} invoking the fluent API of \protect\cref{lst:canvas:generated}},
	label={lst:usage:example}]
using CanvasAPI; using CanvasAPI.FluentAPI;
class UseExample {
  public static void Main(string[] _) {
    // Start.Save().Restore().Draw().Restore().Save().Done<Canvas>(); // Does not compile
    var canvas = Start.Draw().Draw().Save().Draw().Restore()
                      .Draw().Save().Draw().Draw().Done<Canvas>();
    foreach (var token in canvas) System.Console.Write(token + " ");
  }
}
\end{code}

\paragraph*{Treetop in runtime}
Treetop implements the construction in the proof of \cref{theorem:cff:gnf:lower} that
relies on expansive inheritance.
While expansive inheritance is supported in \CSharp, the popular runtime environment~\@.NET prohibits it
\cite[\S{}II.9.2]{cli:05}.
This means that the~\@.NET framework cannot execute programs that invoke Treetop-generated APIs.
To run Treetop programs, one must use an alternative runtime environment that
supports expansive inheritance, e.g., Mono\footnote{See \url{https://www.mono-project.com/}.
Recursively expansive class tables run successfully in version 6.8.0.123.}.

\subsection{Treetop Benchmark}
\label{section:benchmark}
This section describes the results of an experiment that tested the compilation
time of different subtyping machines generated by Treetop.
The experiment resources, including the Treetop executable, scripts, example grammars,
generated subtyping machines, final results, and a detailed description of the
experiment, are available online\footnote{See \url{https://doi.org/10.5281/zenodo.5091711}.}.

Our experiment included four context-free grammars,
shown in \cref{lst:grammar:palindrome,lst:grammar:canvas,lst:grammar:dot,lst:grammar:ambiguous}.
Recall that the line \inline{v ::=} encodes an $\varepsilon$-production, $v \produce \varepsilon$.
\begin{description}
	\item[Palindrome] (\cref{lst:grammar:palindrome})
		This is yet another grammar describing the language of palindromes over $\{a,b\}$.
	\item[Canvas] (\cref{lst:grammar:canvas})
		This context-free grammar describes the Canvas API protocol, discussed in detail
		in \cref{section:use}.
	\item[DOT] (\cref{lst:grammar:dot})
		The DOT DSL is used to draw graphs \cite{Gansner:2000}.
		The deterministic grammar presented here is adapted from \citet{Yamazaki:2019}.
	\item[Ambiguous] (\cref{lst:grammar:ambiguous})
		This grammar describes the inherently ambiguous language \[
			\ell=\{a^nb^mc^md^n~|~n,m \ge 1\}\cup \{a^nb^nc^md^m~|~n,m \ge 1\}
		\]
		discussed by \citet[\S5.4.4]{Hopcroft:Motwani:Ullman:07}.
		Every grammar of an inherently ambiguous context-free language must be ambiguous
		and thus non-deterministic.
\end{description}

\begin{figure}
\begin{tabular}{cc}
	\begin{minipage}{.5\linewidth}
	\begin{code}[backgroundcolor=\color{perfect-grey},caption={Grammar file \texttt{Palindrome.cfg}},label={lst:grammar:palindrome}]
S ::= a S a
S ::= b S b
S ::= a
S ::= b
S ::=
\end{code}
	\end{minipage}
	&
	\begin{minipage}{.45\linewidth}
	\begin{code}[backgroundcolor=\color{perfect-grey},caption={Grammar file \texttt{Canvas.cfg}},label={lst:grammar:canvas}]
Canvas ::= Draw Canvas
Canvas ::= Save Canvas Restore Canvas
Canvas ::= Save Canvas
Canvas ::=
\end{code}
	\end{minipage}
	\\
	\begin{minipage}{.5\linewidth}
	\begin{code}[backgroundcolor=\color{perfect-grey},caption={Grammar file \texttt{DOT.cfg}},label={lst:grammar:dot}]
Graph ::= digraph Statements
Graph ::= graph Statements
Statements ::= Statement Statements
Statements ::=
Statement ::= node Ands NodeAttrs
Ands ::= and Ands
Ands ::=
Statement ::= edge Ands to Ands EdgeAttrs
NodeAttrs ::= NodeAttr NodeAttrs
NodeAttrs ::=
EdgeAttrs ::= EdgeAttr EdgeAttrs
EdgeAttrs ::=
NodeAttr ::= color
NodeAttr ::= shape
EdgeAttr ::= color
EdgeAttr ::= style
\end{code}
	\end{minipage}
	&
	\begin{minipage}{.45\linewidth}
	\begin{code}[backgroundcolor=\color{perfect-grey},caption={Grammar file \texttt{Ambiguous.cfg}},label={lst:grammar:ambiguous}]
S ::= X
S ::= Y
X ::= a X d
X ::= F
Y ::= E G
E ::= a E b
E ::=
F ::= b F c
F ::=
G ::= c G d
G ::=
\end{code}
	\end{minipage}
\end{tabular}
\end{figure}

We used the Treetop command line tool to convert each of the four grammars into
a \CSharp subtyping machine.
Next, we modified the \CSharp programs by adding subtyping queries that invoke the
subtyping machines.
The queries were encoded as dummy variable assignments.
The right-hand side of each query contained types of various sizes, and its contents
were picked at random.
For example, the Palindrome test case could use types \cc{a<a<$\cdot$>{}>} for size $n=2$,
\cc{b<a<b<$\cdot$>{}>{}>} for $n=3$, \cc{a<b<b<a<$\cdot$>{}>{}>{}>} for $n=4$, etc.
Finally, we compiled the subtyping queries on a contemporary laptop using the CSC
compiler\footnote{Windows 10 OS (x64), 16GB RAM, Intel core i7-6700hq processor, CSC version 4.8.4084.0.}
and measured their compilation time.
The results of our experiment are depicted in \cref{figure:Treetop:graph}.
The graph's x-axis shows the size $n$ of the subtyping query, and
the y-axis shows the compilation time of the query in seconds.

\begin{figure}[ht]
	\centering
	\small
	\begin{tikzpicture}[scale=.75]
		\begin{axis}[
			xlabel={Subtyping query size},
			ylabel={Compilation time [s]},
			xmin=-100, xmax=2800,
			ymin=.15, ymax=.8,
			legend pos=outer north east,
		]
		\addplot[green,mark=x] table [col sep=comma] {Ambiguous.csv};
		\addlegendentry{Ambiguous}
		\addplot[red,mark=+] table [col sep=comma] {Canvas.csv};
		\addlegendentry{Canvas}
		\addplot[blue,mark=o] table [col sep=comma] {DOT.csv};
		\addlegendentry{DOT}
		\addplot[black,mark=asterisk] table [col sep=comma] {Palindrome.csv};
		\addlegendentry{Palindrome}
		\end{axis}
	\end{tikzpicture}
	\caption{Compilation times of \protect\CSharp subtyping machines generated by Treetop}
	\label{figure:Treetop:graph}
\end{figure}
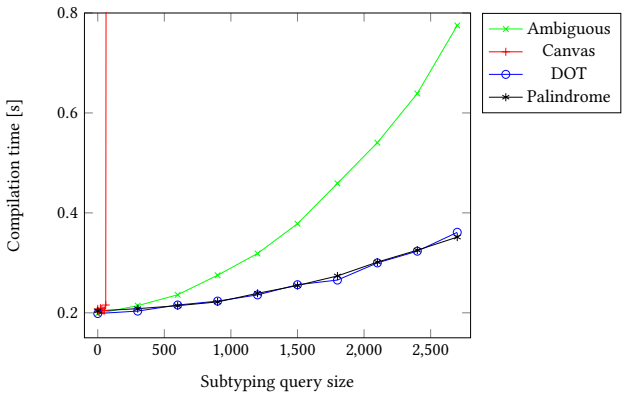

The compilation time of the Ambiguous, DOT, and Palindrome subtyping machines is polynomial
in $n$.
The polynomial trend is expected, considering the correspondence between
covariant subtyping and CFL membership discussed in \cref{section:generator}.
We also see that Ambiguous takes more time to compile than the other two cases,
indicating that different subtyping machines vary in compilation time.
Let us compare these results with those of \citet{Yamazaki:2019} and TypeLevelLR,
a fluent API generator for LR grammars.
The compilation times of TypeLevelLR are linear in $n$, as it only supports deterministic CFLs.
Yet, this linear growth is quite steep:
Using their \emph{expr} API, for instance,
compilation time increases by around 1.5sec between $n=0$ and $n=200$ in \NonCitingUse{Scala},
by 500ms in \NonCitingUse{Haskell}, and by 15sec in \NonCitingUse{CC}.
Therefore, while the compilation times of Treetop are expected to pass those of TypeLevelLR
as the input becomes longer, they are actually much lower for ``small'' inputs ($n\le 2700$, and probably beyond that).
To put this into perspective, in a recent empirical study, \citet{Nakamaru:2020b} found that only
4.98\% of the repositories in their data set use fluent chains with 42 method calls or more.

Nevertheless, the forth test case, Canvas, produced different results.
The Canvas subtyping query of size $n=70$ took 2.8sec to compile, around 2.6sec more
than the previous query of size $n=60$.
This query, generated at random, contained 33 instances of \texttt{Save} and \texttt{Restore},
the critical tokens of the Canvas grammar.
The seemingly exponential growth in compilation time of Canvas suggests that
some subtyping queries generated by Treetop may not be practical.
The difference between the compilation times of Canvas and the other three grammars
originate from the subtyping algorithm implemented in the CSC compiler.
The high compilation time could be caused by a bug, similar to the Javac compiler bug found by \citet{Gil:19}.
It is also possible, however, that the CSC compiler needs to be specialized to optimally resolve covariant subtyping.


	\section{Discussion}
	\label{section:zz}
	\label{section:discussion}
	A program that uses the subtyping algorithm to conduct a computation is called
a subtyping machine. By formalizing subtyping machines as
forest recognizers, we tie the decidable subtyping type systems of Kennedy and Pierce with familiar
families of tree grammars. We proved that non-expansively-recursive subtyping
machines recognize regular forests (\cref{theorem:regf}) and non-contravariant
machines recognize context-free forests with GNF
(\cref{theorem:cff:gnf}).  It was  also shown that multiple instantiation
inheritance corresponds to grammar non-determinism
(\cref{theorem:dregf,theorem:dcff}). The results are depicted in
\cref{figure:lattice} in \cref{section:aa}.

In non-expansive subtyping, all infinite proofs are cyclic, so they can be
detected by recording all visited sub-queries. The construction in the
proof of \cref{theorem:regf:upper} encodes cyclic proofs as
recursive productions~$v \produce v$---which are immediately removed from the
grammar. This means that infinite proofs can be intercepted in
pre-computation, i.e., even before the subtyping query is inspected.
Our construction, however, assumes that one of the types
in the query is known and fixed, following our definition of the language
of a class table (\cref{definition:class:table:language}). In the
future, our method may be developed into a better algorithm for non-expansive
subtyping that does not record sub-queries.

Context-free tree grammars (CFTGs) are recognized by pushdown tree automata
(PDTA), that use an \emph{auxiliary tree store} and accept a \emph{tree
input}~\cite{Guessarian:83}.  The proof of \cref{theorem:cff:gnf:upper}
describes a reduction of non-contravariant subtyping machines to CFTGs. In
contrast to the regular case, here the left-hand side of the subtyping query~$t_\bot$ is encoded into the
grammar initial tree form~$t₀$ and is not involved with grammar productions.
Therefore, we can resolve non-contravariant subtyping using modified PDTAs,
which accept a second input tree and use this as their initial store. Given
query~$t_\bot \subtype t_\top$, we assign~$t_\bot$ to the tree store (possibly after the
transformation of \cref{eq:covariant:encoding}) and run the PDTA on
input~$t_\top$. As our construction yields grammars in GNF, their
corresponding PDTAs run in real time~\cite{Guessarian:83}.
If multiple instantiation inheritance is also restricted, then the grammars
become deterministic and so does the PDTA run time. Thus, subtyping in~\Txm
can be resolved by a PDTA in non-deterministic real time, and in~\Tx in
deterministic real time.

The theoretical results are accompanied by a POC software artifact, Treetop.
Through Treetop, we demonstrate that subtyping can be used to implement
useful and powerful metaprogramming applications even in a decidable type system.
Treetop converts context-free grammars into \CSharp subtyping machines
and can be used to embed context-free DSLs as fluent APIs;
its practical relevance can be argued based on the work of
\citet{Ferles:20}, who found cases of context-free API protocols in the wild.
Comparing Treetop to the fluent API generator of \citet{Grigore:2017}, we see that
both support all CFGs, but Treetop does so in a decidable type system
fragment.
While Grigore's generator constructs very large subtyping machines that take
a very long time to compile, Treetop generates small machines that compile
relatively fast (in some cases).
In theory, Treetop's subtyping machines run in polynomial time.
In practice, our experiment found an example that leads to possibly
exponential compilation time using the CSC compiler.
In addition, Treetop's fluent API generation does not yet match the qualities
of previous practical generators (e.g.,~\cite{Nakamaru:17,Gil:19,Yamazaki:2019}).
More research work (but also implementation) is required to support features
such as auto-completion, early failure, and, run-time parsing to create an AST.

\paragraph*{A note on mixin inheritance}
Kennedy and Pierce included mixin inheritance in their abstract type system, allowing classes
to inherit their type variables,~$γ(x):x$.
Mixin inheritance was not discussed in this paper, given that not many
programming languages actually support it.
As it turns out, most of our theorems hold when this feature
is re-introduced to the type system.
The construction in the proof of \cref{theorem:regf:upper} is left unchanged,
since it relies on the non-expansive guarantee of inheritance in~\Tcm and does
not depend on the structure of inheritance declarations.
In contrast, the proof of \cref{theorem:cff:gnf:upper} breaks when confronted with mixin
inheritance. The construction encodes (transitive) inheritance
relations~$γ(\vv{x}):^*σ(\vv{τ})$ as context-free productions in GNF,~$γ(\vv{x}) \produceσ(\vv{τ})$.
The production corresponding to mixin
inheritance is an~$𝜀$-production,~$γ(\vv{x}) \produce xᵢ$, which does not conform
to GNF.
In this case the resulting grammar describes a general context-free forest,
so the expressiveness of~\Txm is now bounded from above by class NCF instead
of NCF\sub{GNF} (which is the lower bound).
The computational power of type system~\Tx, however, is already
bounded from below by general deterministic NCF, so it too stays the same.
Therefore, the introduction of mixin inheritance to the subtyping-based type
systems in \cref{figure:lattice} only adds a loose bound to~\Txm,~\[ 
  \|NCF\sub{GNF}|⊆𝓛(\Txm)⊆\|NCF|. 
\]

	\bibliographystyle{ACM-Reference-Format}
	\bibliography{00,author-names,publishers,big}
	
\end{document}